\documentclass[a4paper]{article}

\usepackage{amsthm}

\usepackage[margin=1.5in,footskip=0.5in]{geometry}
\usepackage[notransparent]{svg}
\newtheorem{defin}{Definition}
\usepackage{tikz}
\usepackage{pgfplots}
\usepackage{xcolor}
\usetikzlibrary{topaths,calc}
\usepackage{makecell}
\usepackage[numbers]{natbib} 
\usepackage{graphicx}
\usepackage{subcaption}
\usepackage{amsmath, amsthm, amssymb, mathrsfs, amsfonts}
\usepackage{transparent}
\usepackage{booktabs}
\newtheorem{Prop}{Proposition}

\newtheorem{remark}{Remark}
\newtheorem{Lemma}{Lemma}

\usepackage{authblk}
\usepackage{soul}

\newcolumntype{?}{!{\vrule width 1pt}}

\title{Chemically inspired Erd\H{o}s-R\'enyi oriented hypergraphs}

\author[$^{1,2}$]{Angel Garcia-Chung}
\author[$^{1,2}$]{Marisol Berm\'udez-Monta\~na}
\author[$^{2,3,4,5,6,7}$]{Peter F. Stadler}
\author[$^{2,6,8}$]{J\"urgen Jost}
\author[$^{2,3,9},^*$]{Guillermo Restrepo}

\affil[$^{1}$]{Tecnol\'ogico de Monterrey, Escuela de Ingenier\'ia y Ciencias, Estado de M\'exico 52926, M\'exico}
\affil[$^{2}$]{Max Planck Institute for Mathematics in the Sciences, Leipzig, Germany}
\affil[$^{3}$]{Interdisciplinary Center for Bioinformatics Universit\"at Leipzig, Leipzig, Germany}
\affil[$^{4}$]{Bioinformatics Group, Department of Computer Science Universit\"at Leipzig, Leipzig, Germany}
\affil[$^{5}$]{Institute for Theoretical Chemistry, University of Vienna, Vienna, Austria}
\affil[$^{6}$]{The Santa Fe Institute, Santa Fe, New Mexico, USA}
\affil[$^{7}$]{Facultad de Ciencias, Universidad Nacional de Colombia, Bogot\'a, Colombia}
\affil[$^{8}$]{Center for Scalable Data Analytics and Artificial Intelligence,  Universit\"at Leipzig, Leipzig, Germany}
\affil[$^{9}$]{School of Applied Sciences and Engineering, EAFIT University, Medell\'in, Colombia}
\affil[$^*$]{Corresponding author: Guillermo Restrepo, restrepo@mis.mpg.de}

\begin{document}
\maketitle
\begin{abstract}
  \small{ High-order structures have been recognised as suitable models for systems going beyond the binary relationships for which graph models are appropriate.  Despite their importance and surge in research on these structures, their random cases have been only recently become subjects of interest.  One of these high-order structures is the oriented hypergraph, which relates couples of subsets of an arbitrary number of vertices.  Here we develop the Erd\H{o}s-R\'enyi model for oriented hypergraphs, which corresponds to the random realisation of oriented hyperedges of the complete oriented hypergraph.  A particular feature of random oriented hypergraphs is that the ratio between their expected number of oriented hyperedges and their expected degree or size is 3/2 for large number of vertices.  We highlight the suitability of oriented hypergraphs for modelling large collections of chemical reactions and the importance of random oriented hypergraphs to analyse the unfolding of chemistry. }\\
Keywords: graphs, hypergraphs, chemical space, random model, Erd\H{o}s-R\'enyi.

\end{abstract}
\section{Introduction}
Graphs are often selected as the mathematical structure to analyse binary relationships between objects of a system \cite{Schmidt2011,Diestel2018}.  As in any mathematical setting, it is important to determine the bounds of the structures modelling the systems.  This allows for determining how far or close a system is from its theoretical extremes, which leads to study the lower and upper bounds of systems as well as its random cases.  In graph theory, this amounts to determine the edge-less, complete and random graphs.  The edge-less graph corresponds to a system devoid of relationships, while the complete graph to a system depicting the maximum number of relationships.  In turn, the random graph corresponds to a system whose relationships are randomly assigned.

For the sake of clarity and for setting up the notation, a \emph{graph} $G=(V,E)$ corresponds to a set of \emph{vertices} $V$ and a set of \emph{edges} $E$ ($E\subseteq \{\{x,y\}: x,y \in V\}$). 
Note that we consider simple graphs, that is, graphs without multiple edges and without loops.  Thus, $E$ is a set, but not a multiset, and if $\{x,y\}\in E$, it follows that $x \neq y$. While the edge-less and complete graphs are straightforwardly defined, the former as a graph $G$ with an empty set $E$ and the second as one with $n(n-1)/2$ edges, the random graph allows for several approximations.

The earliest and most general random graph model was reported by Erd\H{o}s and Rényi in 1959 \cite{ErdosRenyi1959} and is constructed by the random assignment of edges on the set $V$, with probability of assignment $p$.  That is, given a set $V$ of $n$ vertices, the random graph corresponds to the realisation, or not, of every possible edge on $V$.  The probability of realisation of the edges is given by $p$.  This can be thought of as the result of an algorithm that takes every possible couple of vertices in $V$ and decides whether to link them or not.  The decision depends on generating a random number, and if that number happens to be less than a predetermined value, denoted as $p$, then the vertices are connected; otherwise, they remain disconnected. 

Despite the widespread use of graphs, it has been acknowledged that certain systems exhibit relationships that surpass the typical binary relations represented by graphs \cite{Boccaletti2023,Chodrow2020}.  These high-order relations refer to $k$-ary relationships, which involve interactions among $k$ or less vertices, where $k>2$.  Examples of these systems include functional and structural brain networks \cite{Yixin2021,Shi2017}, protein interaction networks \cite{Murgas2022}, chemical reaction networks \cite{Restrepo2022,Leal2021,Stadler2015,Stadler2018}, as well as semantic \cite{Menezes2021,Zuguo2017} and co-authorship networks \cite{Xie2021}.  For instance, in co-authorship networks, a high-order relationship, say of order five, corresponds to the five authors of a five-author paper.

Mathematical settings to address high-order relations include simplicial complexes and hypergraphs \cite{Mulas2022}.  The former are selected for cases where nested sets of related vertices are relevant and the latter for cases where arbitrary sets of vertices are the focus.  Hence, hypergraphs are more general than simplicial complexes and have been the subject of recent studies upon their properties and applications \cite{Mulas2019,Eidi2020,Eidi2020a,Mulas2020a,Sun2021,Thakur2009,Ausiello2017}.

One of the applications of hypergraphs is for modelling chemical reactions, which becomes the leading example and motivation of the current paper.  In this chemical setting, hypergraphs require to encode binary relationships among sets of vertices (substances).  Despite the use of hypergraphs in chemistry, the extreme cases of chemical hypergraphs have not been yet studied.  Here we report extreme cases of chemical hypergraphs.  That is, the mathematics of edge-less, complete and random chemical hypergraphs, as well as some of their properties.  Before discussing these structures, we provide some details on their use for modelling chemical reactions.

\section{Graphs and hypergraphs for modelling the chemical space}

Chemical space spans all substances and reactions reported over the history of chemistry \cite{Llanos2019,Restrepo2022} and has been initially modelled with graphs \cite{Fialkowski2005,Grzybowski2006}.  In this setting, reactions in Figure \ref{fig:modles_rxns}a are modelled as graphs as shown in Figure \ref{fig:modles_rxns}b.\footnote{Models of the chemical space strongly emphasise the role of reactions as the ``gluing'' aspect relating substances, which, in turn, endows the set of substances with a structure.  This is what actually turn the set of substances into a space \cite{Restrepo2022}.  In such a setting, however, non-connected substances, often arising from chemical extractions, play an important role in the space, as they represent new non-synthetic chemicals, which may, or not, remain disconnected in the space, or which require a certain amount of time to be connected to the network.  This was the case of the noble gases, for instance, which, for many years, remained as isolated substances of the chemical space.  Moreover, determining the average time required to connect a substance to the network of the chemical space is of central importance for studies on the evolution of chemical knowledge, as well for the chemical industry \cite{JostRestrepo2022}.}  The model encodes the binary relationship between an educt and a product of a reaction.  Although informative, the graph model misses important chemical information.  For instance, whether a substance can be reached from another in a given chemical space.  This is the case of A and E (in Figure \ref{fig:modles_rxns}b), which are connected via two edges: $\{$A, F$\}$ and $\{$F, E$\}$. Nevertheless, Figure \ref{fig:modles_rxns}a shows that from A, E cannot be obtained.  This shortcoming is solved by adding more structure to the graph, namely by adding direction to the edges and by modelling reactions as \emph{directed graphs} (Figure \ref{fig:modles_rxns}c).  This model now shows that E cannot be reached from A, but that G can be reached from A.  Despite the improvements in incorporating the direction of the chemical transformation, the directed graph model still lacks other important chemical aspects, namely that the model does not inform which substances react to produce other chemicals.  This shortcoming is solved by modelling reactions as \emph{hypergraphs}.  Figure \ref{fig:modles_rxns}d shows the hypergraph for the reactions of Figure \ref{fig:modles_rxns}a, which encodes the different sets of substances in the reactions of the chemical space analysed.  Those sets correspond to actual mixtures of substances either before (educts) or after the reaction has taken place (products).  From Figure \ref{fig:modles_rxns}d is clear that A is found as mixed with B, as well as E with D and A; and D with C.  Likewise, the hypergraph shows that F and G are substances whose transformations do not require any other substance of the chemical space.

\begin{figure}[ht!]
\centering
\includegraphics[width=0.8\linewidth]{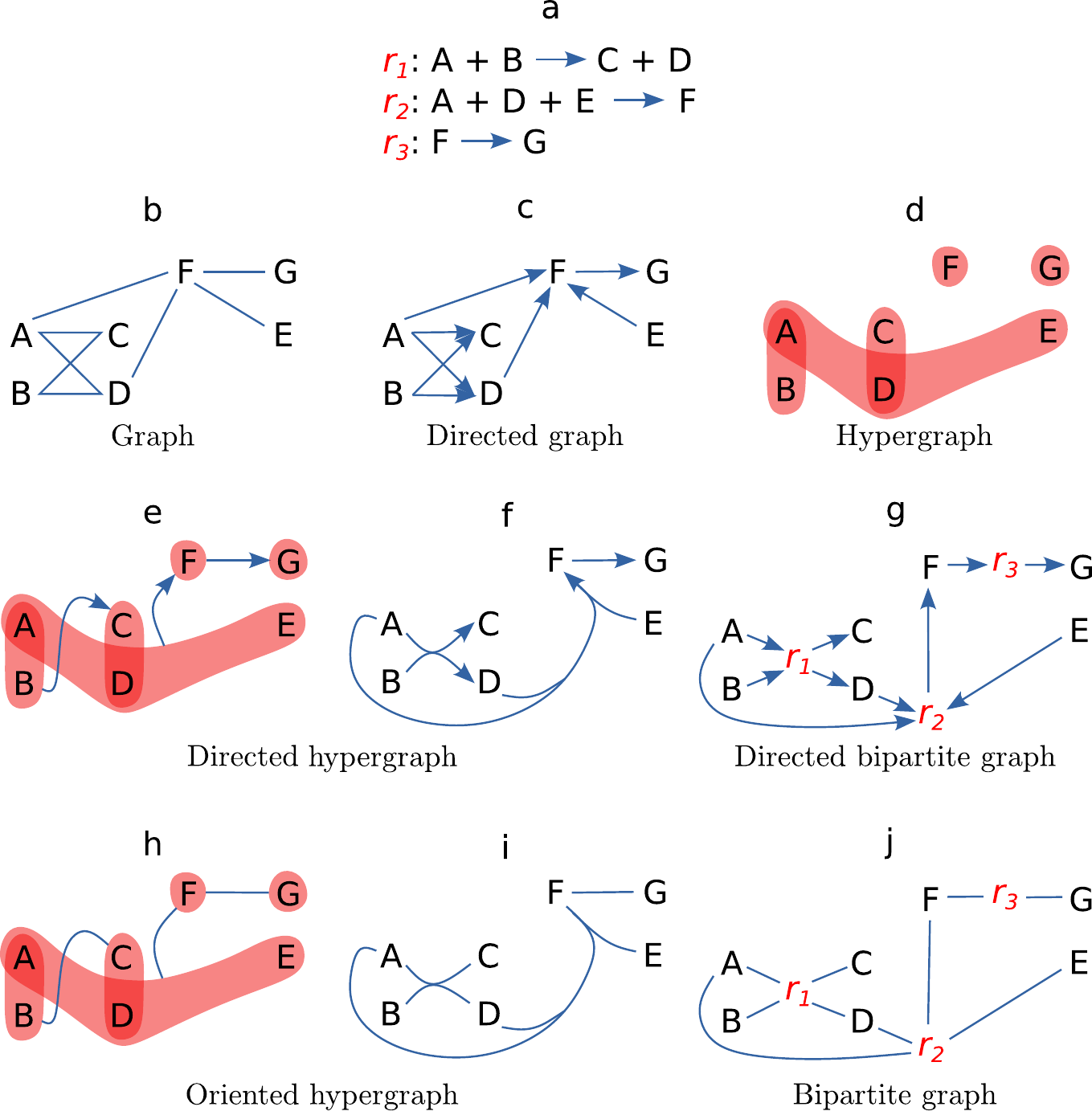}
\caption{Chemical reactions as graphs and hypergraphs.  All structures, from b to j correspond to (hyper)graph models for the chemical reactions in a, which constitute a chemical space of seven substances and three reactions.}
\label{fig:modles_rxns}
\end{figure}

As chemical reactions are actually directed binary relationships among sets of substances, that is among educts and products, a further refinement of the model requires introducing this binary relation, which is encoded by \emph{directed hypergraphs}.  Figure \ref{fig:modles_rxns}e illustrates how the directed hypergraph encodes the transformation of the set of educts \{A, B\} into the products \{C, D\}, as well as the reaction of \{A,D,E\} to produce substance F.  Likewise, it shows the rearrangement of F into G.

Alternative representations of the directed hypergraph of Figure \ref{fig:modles_rxns}e are shown in Figures \ref{fig:modles_rxns}f and g.  Figure \ref{fig:modles_rxns}f, besides emphasising the directed relationship among educts and products, highlights the role of substances (vertices) in the transformation.  Figure \ref{fig:modles_rxns}g maps the directed hypergraph back to the realm of directed graphs.\footnote{In formal terms, it can be considered as a mapping to the category of directed graphs.}  This time not to the directed graphs of Figure \ref{fig:modles_rxns}c but rather to \emph{directed bipartite graphs}, where besides the usual vertices representing substances, a new set of vertices is introduced, namely that representing reactions.  A relaxed version of structures in Figures \ref{fig:modles_rxns}e, f and g are the corresponding undirected structures, shown in Figures \ref{fig:modles_rxns}h, i and j, which are different representations of the associated \emph{oriented hypergraph}. Note how oriented hypergraphs constitute a suitable model for reversible reactions.  Thus, for instance, the oriented hypergraph \{A, B\}-\{C, D\} encodes the reactions A + B $\rightarrow$ C + D, as well as C + D $\rightarrow$ A + B.\footnote{Formally, and in particular following the notation described in Definition  \ref{def:chemical_ultragraph}, \{A, B\}-\{C, D\} can be written down as \{\{A, B\},\{C, D\}\}.}  This encoding is chemically sound, as every reaction is intrinsically reversible.  The actual direction observed in wet-lab experiments arises from the energetic conditions in which molecules are embedded in the reaction process.\footnote{A fairly concise mathematical description on the topic is found in \cite{Mueller:22a}.  It may also happen that the oriented hyperedge \{A, B\}-\{C, D\}, beyond encoding A + B $\rightarrow$ C + D and/or C + D $\rightarrow$ A + B as two reactions occurring at particular conditions, also encodes the same transformation but occurring under different conditions and even through different reaction mechanisms.}

It is upon the oriented hypergraph model for chemical reactions that we study its extreme cases and develop an Erd\H{o}s-Rényi random model.  In the next section the mathematical elements setting up the stage for such study are presented.

\section{Chemically inspired oriented hypergraphs}

As discussed in the previous section, the most informative model for the chemical space is the directed hypergraph.  Nevertheless, for the sake of generality, we report in the current paper results for extreme cases and an Erd\H{o}s-Rényi model for oriented hypergraphs.  That is, in what follows, we regard the chemical space as devoid of direction in its chemical reactions and we focus only on the connectivity of the substances via reactions, while preserving the important aspect of chemical reactions of relating sets of substances.

We introduce some definitions, which assume a fixed number $n$ of vertices gathered on the set $V$.  Upon $V$, subsets of vertices are defined, which are pair-wise related by the oriented hypergraph.  These subsets gather together substances appearing as either educts or products of reactions in the chemical space.  Chemical reactions can be classified as either catalytic or non-catalytic.  The former involve the use of a catalyst, which is a substance added to the educts to speed up the synthesis of reaction products.  Catalysts are not consumed in the reaction, which distinguishes them form the educts.  Chemical notation encodes the catalyst as a label of the reaction.  If, for instance, A + B $\rightarrow$ C + D is catalysed by E, the reaction is written down as A + B $\xrightarrow{\text{E}}$ C + D.  Otherwise, if there is no catalyst involved, A + B $\rightarrow$ C + D represents the reaction.  In this classification, autocatalytic reactions constitute a particular case of catalytic reactions, where at least one of the educts acts as a catalyst.  Hence, A + B $\rightarrow$ B + C is an example of an autocatalytic reaction,\footnote{Note that stoichiometric coefficients are disregarded in this notation.} which can be considered as the sequence of two reactions: first A + B $\rightarrow$ Z, followed by Z $\rightarrow$ B + C, where Z is known as the reaction intermediate.  Hence, oriented hypergraphs turn to be suitable models for all chemical reactions.  Therefore, we model the chemical space as discussed in Definition \ref{def:chemical_ultragraph}.

\begin{defin} \label{def:chemical_ultragraph}
A \emph{chemical space} of $n$ substances gathered in $V$ is modelled as an \emph{oriented hypergraph} $G=(V,E)$, with oriented hyperedges (reactions) gathered in $E\subseteq \{\{X,Y\}:X,Y\in {\cal P}(V)\setminus \{\varnothing \} \text{ and } X\cap Y=\varnothing\}$. $X$ and $Y$, which are sets of substances, are called \emph{hypervertices} of the chemical space and every \emph{oriented hyperedge} $r\in E$ is called a \emph{chemical reaction} of the chemical space. 
\end{defin}

Importantly, in our framework, substances consumed or produced in a chemical reaction are restricted to be in the set $V$. Moreover, hypervertices cannot be empty (Definition \ref{def:chemical_ultragraph}) because there is no chemical reaction leading or starting from an empty set of substances.  Likewise, a reaction cannot start from the complete set of substances, as there would be no room for synthesising a new substance.  Similarly, as no reaction can lead to the set containing all substances, the hypervertex containing all vertices is disregarded from the model.\footnote{Moreover, in this setting we are disregarding the particular reaction conditions at which reactions are carried out. Nonetheless, they can be incorporated as labels of oriented hyperedges.}

Therefore, the maximum number of hypervertices in a chemical space is given by $2^n-2$.  They correspond to the maximum number of combinations of substances chemists can try, which may lead to another set of substances within the given chemical space.\footnote{This upper bound holds significance in, for instance, research on the origin of life.  A mathematical setting for such studies is provided by Dittrich's chemical organisation theory \cite{Dittrich2007}, where finding sequences of reactions involving a given subset of substances of the chemical space is an important aspect of the approach.}  Now we classify those sets of substances by the number of substances they contain.

Let $V$ be the set of $n$ vertices (substances) and  $B_a$ the set gathering together subsets of $V$ of $a$ vertices.  Thus, $B_1$ collects all substances ($B_1=V$), $B_2$ all couples of substances and so forth. The complete set of hypervertices $B$ is given by
\begin{align}
B = \bigcup_{a = 1}^{n-1} B_a = {\cal P}(V)\setminus \{\varnothing, B_n\}.
\end{align}

The number of hypervertices of size $a$ corresponds to the cardinality of $B_a$, which is given by the number of combinations of $a$ vertices that can be obtained out of $n$ vertices. Thus,
\begin{align}
|B_a|={\cal C}^n_a = {n \choose a}. \label{eq:B_i}
\end{align}

As $B$ contains all possible sets of vertices (hypervertices of the oriented hypergraph) involved in chemical reactions for a given set of vertices, a suitable object gathering information on the connectivity of these hypervertices, that is on the chemical reactions, is the generalised \emph{adjacency matrix of the oriented hypergraph} ${\bf M} = [M_{i,j}]_{2^n-2 \times 2^n-2}$, where the indices $i,j = 1,2,\dots, 2^n - 2$ run over all the possible hypervertices for a given $n$. The components of the adjacency matrix are given as
\begin{align}
M_{i,j}=\left\{ \begin{array}{cl}
1 & \mbox{if} \hspace{.1cm} r = \{ b_i , b_j \} \in E, \\ 
0 & \mbox{otherwise.}
\end{array}\right.
\label{adjmatrix}
\end{align}

Thus, any 1-entry of ${\bf M}$ corresponds to a chemical reaction between the hypervertex $b_i$ that gathers $i$ substances and the hypervertex $b_j$ that gathers $j$ substances. Note that ${\bf M}$ is symmetric ($M_{i,j}=M_{j,i}$) because the reactions $b_i \rightarrow b_j$ and $b_j \rightarrow b_i$ are equivalent in the oriented hypergraph. Any 0-entry in ${\bf M}$ indicates either that the reaction is possible but not yet realised in the chemical space or that the reaction is not possible.  In the first case, the two hypervertices $b_i$ and $b_j$ may be connected by a chemical reaction, but the chemical space at disposal has not realised the reaction.  In the second case, there is at least a common substance between $b_i$ and $b_j$ and the reaction is not possible in our scheme. 

Let us consider a toy-chemical space of four reactions over the set of substances $V=\{$A, B, C, D$\}$ (Figure \ref{fig:toy_model}), whose corresponding generalised matrix is shown below.

\begin{table}[!h]
\centering
\caption{Adjacency matrix ${\bf M}$ for a chemical space of four substances \{A, B, C, D\}.  1-entries correspond to realised reactions, while 0-entries to either possible reactions (black) or to impossible reactions (red).  This latter correspond to reactions with at least a common substance in the set of educts and products.  Matrix blocks, separated by bold lines, gather together sets of educts with cardinality $i$ and sets of products with cardinality $j$.} \label{AMatrixN4}
\begin{tabular}{l?c|c|c|c?c|c|c|c|c|c?c|c|c|c?}
& A & B & C & D & AB & AC & AD & BC & BD & CD & ABC & ABD & ACD & BCD \\ 
\bottomrule[1pt]
A & {\color{red}{0}} & {\bf 1} & 0 & 0 & {\color{red}{0}} & {\color{red}{0}} & {\color{red}{0}} & 0 & 0 & 0 & {\color{red}{0}} & {\color{red}{0}} & {\color{red}{0}} & 0\\
\hline
B & {\bf 1} & {\color{red}{0}} & 0 & 0 & {\color{red}{0}} & 0 & 0 & {\color{red}{0}} & {\color{red}{0}} & 0 & {\color{red}{0}} & {\color{red}{0}} & 0 & {\color{red}{0}} \\
\hline
C & 0 & 0 & {\color{red}{0}} & 0 & 0 & {\color{red}{0}} & 0 & {\color{red}{0}} & 0 & {\color{red}{0}} & {\color{red}{0}} & 0 & {\color{red}{0}} & {\color{red}{0}}\\
\hline
D & 0 & 0 & 0 & {\color{red}{0}} & 0 & {\bf 1} & {\color{red}{0}} & {\bf 1} & {\color{red}{0}} & {\color{red}{0}} & 0 & {\color{red}{0}} & {\color{red}{0}} & {\color{red}{0}} \\
\bottomrule[1pt]
AB & {\color{red}{0}} & {\color{red}{0}} & 0 & 0 & {\color{red}{0}} & {\color{red}{0}} & {\color{red}{0}} & {\color{red}{0}} & {\color{red}{0}} & 0 & {\color{red}{0}} & {\color{red}{0}} & {\color{red}{0}} & {\color{red}{0}}\\
\hline
AC & {\color{red}{0}} & 0 & {\color{red}{0}} & {\bf 1} & {\color{red}{0}} & {\color{red}{0}} & {\color{red}{0}} & {\color{red}{0}} & 0 & {\color{red}{0}} & {\color{red}{0}} & {\color{red}{0}} & {\color{red}{0}} & {\color{red}{0}}\\
\hline
AD & {\color{red}{0}} & 0 & 0 & {\color{red}{0}} & {\color{red}{0}} & {\color{red}{0}} & {\color{red}{0}} & {\bf 1} & {\color{red}{0}} & {\color{red}{0}} & {\color{red}{0}} & {\color{red}{0}} & {\color{red}{0}} & {\color{red}{0}}\\
\hline
BC & 0 & {\color{red}{0}} & {\color{red}{0}} & {\bf 1} & {\color{red}{0}} & {\color{red}{0}} & {\bf 1} & {\color{red}{0}} & {\color{red}{0}} & {\color{red}{0}} & {\color{red}{0}} & {\color{red}{0}} & {\color{red}{0}} & {\color{red}{0}}\\
\hline
BD & 0 & {\color{red}{0}} & 0 & {\color{red}{0}} & {\color{red}{0}} & 0 & {\color{red}{0}} & {\color{red}{0}} & {\color{red}{0}} & {\color{red}{0}} & {\color{red}{0}} & {\color{red}{0}} & {\color{red}{0}} & {\color{red}{0}}\\
\hline
CD & 0 & 0 & {\color{red}{0}} & {\color{red}{0}} & 0 & {\color{red}{0}} & {\color{red}{0}} & {\color{red}{0}} & {\color{red}{0}} & {\color{red}{0}} & {\color{red}{0}} & {\color{red}{0}} & {\color{red}{0}} & {\color{red}{0}}\\
\bottomrule[1pt]
ABC & {\color{red}{0}} & {\color{red}{0}} & {\color{red}{0}} & 0 & {\color{red}{0}} & {\color{red}{0}} & {\color{red}{0}} & {\color{red}{0}} & {\color{red}{0}} & {\color{red}{0}} & {\color{red}{0}} & {\color{red}{0}} & {\color{red}{0}} & {\color{red}{0}}\\
\hline
ABD & {\color{red}{0}} & {\color{red}{0}} & 0 & {\color{red}{0}} & {\color{red}{0}} & {\color{red}{0}} & {\color{red}{0}} & {\color{red}{0}} & {\color{red}{0}} & {\color{red}{0}} & {\color{red}{0}} & {\color{red}{0}} & {\color{red}{0}} & {\color{red}{0}}\\
\hline
ACD & {\color{red}{0}} & 0 & {\color{red}{0}} & {\color{red}{0}} & {\color{red}{0}} & {\color{red}{0}} & {\color{red}{0}} & {\color{red}{0}} & {\color{red}{0}} & {\color{red}{0}} & {\color{red}{0}} & {\color{red}{0}} & {\color{red}{0}} & {\color{red}{0}}\\
\hline
BCD & 0 & {\color{red}{0}} & {\color{red}{0}} & {\color{red}{0}} & {\color{red}{0}} & {\color{red}{0}} & {\color{red}{0}} & {\color{red}{0}} & {\color{red}{0}} & {\color{red}{0}} & {\color{red}{0}} & {\color{red}{0}} & {\color{red}{0}} & {\color{red}{0}}\\
\bottomrule[1pt]
\end{tabular}
\vspace*{-4pt}
\end{table}

\begin{figure}[ht!]
\centering
\includegraphics[width=0.6\linewidth]{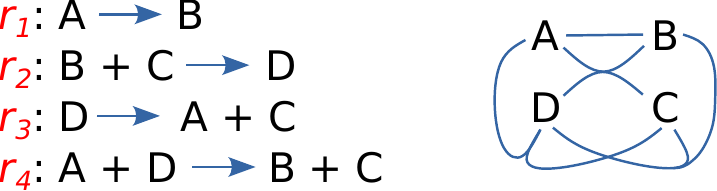}
\caption{Toy chemical space constituted by four substances \{A,B,C,D\} and four reactions $r_i$.  On the left, reactions are presented in chemical notation and on the right the chemical space is depicted as an oriented hypergraph.}
\label{fig:toy_model}
\end{figure}
\newpage
Note, for instance, that the reaction A $\rightarrow$ B or B $\rightarrow$ A is part of the chemical space gathered in ${\bf M}$ as $M_{\text{A,B}}=M_{\text{B,A}}=1$ (Table \ref{AMatrixN4}).  In contrast, A $\rightarrow$ C or C $\rightarrow$ A, although a possible reaction, has a 0-entry in ${\bf M}$ because it is not a part of the chemical space (Figure \ref{fig:toy_model}).  Reaction A $\rightarrow$ AB or AB $\rightarrow$ A, which correspond to A  $\rightarrow$ A + B or A + B $\rightarrow$ A, in chemical notation, is not possible because of the commonality of A in both hypervertices, therefore it is shown as a 0-entry in ${\bf M}$.  We distinguish two kinds of 0-entries in ${\bf M}$.  Those in black font correspond to possible but not realised reactions.  Those in red to impossible reactions because of commonality of at least a substance between educts and products.\footnote{The number of black 0-entries amounts to the unrealised chemical reactions, which together with the 1-entries correspond to the potential chemical space, as called by some philosophers of chemistry \cite{Schummer1996}.}

${\bf M}$ can be arranged by the number of vertices belonging to the hypervertices.  That is, column and rows in ${\bf M}$ can be arranged from $B_1$, $B_2$, $\ldots$ until $B_{n-1}$.  This is the scheme we adopted to present ${\bf M}$ above (Table \ref{AMatrixN4}).  The number of vertices of the hypervertices allows for classifying reactions in terms of their size.  Given a reaction $r=\{b_i,b_j\}$ between hypervertices $b_i$ and $b_j$, the size of $r$ is given by $s(r)=i+j$.  That is, the \emph{size of a reaction} corresponds to the number of substances involved in the reaction.\footnote{The size of a reaction corresponds to the molecularity of the reaction \cite{Atkins2023} if the stoichiometric coefficients of the reaction are regarded.  As this is not, in general, the case in studies on the chemical space \cite{Llanos2019,Restrepo2022}, the size of a reaction may be regarded as a proto-molecularity of the reaction.  It only accounts for the number of different chemicals reported in the reaction, but not for their actual figures.  Often, chemists omit writing, for instance, water or carbon dioxide, as those substances can be inferred from the context of the reaction or because of the tradition to emphasise the target product of a reaction, namely of a chemical synthesis \cite{JostRestrepo2022}.}  Thus, $s(r_1)=2$, $s(r_2)=s(r_3)=3$ and $s(r_4)=4$ for the chemical space of Figure \ref{fig:toy_model}.  Reaction size is bounded by $2\leq s(r)\leq n$ as a reaction must involve at least an educt and a product and the largest reaction must involve no more than $n$, the number of substances of the chemical space.

Based on the size of reactions, the size of the chemical space, encoded in $G$ (Definition \ref{def:chemical_ultragraph}), can be introduced.
\begin{defin}\label{def:size_CS}
The \emph{size} $s(G)$ of an oriented hypergraph $G$, whose oriented hyperedges are gathered in $E$, is given by
\begin{align}
s(G)=\sum_{r\in E}s(r).
\end{align}
\end{defin}
Note how the chemical space of four substances in Figure \ref{fig:toy_model} has $s(G)=12$.  As we discuss below, $s(G)$ becomes a proxy for the connectivity of the chemical space, which is straightforward approached through the degree of a vertex (substance).

Following the definition of vertex degree in graph theory \cite{Diestel2018}, the degree of a vertex $v \in V$ ($d(v)$) of an oriented hypergraph $G$ corresponds to the number of oriented hyperedges (reactions) in which the substance participates.  For the oriented hypergraph of Figure \ref{fig:toy_model}, $d(\text{A})=d(\text{B})=d(\text{C})=d(\text{D})=3$.  Likewise, the degree of an oriented hypergraph $d(G)$ can be defined.
\begin{defin}\label{def:degree_CS}
The \emph{degree} $d(G)$ of an oriented hypergraph $G$, of vertices gathered in $V$, is given by
\begin{align}
d(G)=\sum_{v\in V}d(v).
\end{align}
\end{defin}
Thus, $d(G)=12$ for the oriented hypergraph of Figure \ref{fig:toy_model}. There is an interesting relation between size and degree of a chemical space modelled as an oriented hypergraph.

\begin{Lemma}\label{prop:size_equal_degree}
The size of an oriented hypergraph $G$ and its degree are equal. That is
\begin{align}
s(G) = d(G).
\end{align}
\begin{proof}
Given an oriented hypergraph $G=(V,E)$ made of $n$ vertices (substances) gathered in $V$ and of $u$ reactions (oriented hyperedges) gathered in $E$, $G$ is equivalent to a bipartite graph whose vertices correspond to both substances $(v)$ and reactions ($r$).  As the degree sum formula for a bipartite graph is \cite{lovasz1993}
\begin{align}
d(G) = \sum_{v \in V} d(v) = \sum_{r \in E} d(r) = u,
\end{align}
then $d(G) = u$ and as $s(r)=d(r)$, then
\begin{align}
\sum_{r \in E} d(r) = \sum_{r \in E} s(r) = s(G) =u.
\end{align}
\end{proof}
\end{Lemma}

Thus, size and degree of the oriented hypergraph $G$ modelling the chemical space indicate how tight or dense the chemical space is.  Low size or degree values indicate a sparse chemical space, while high values a strongly connected space.  In order to provide a baseline for comparison of sizes and degrees of chemical spaces, upper and lower bounds of those oriented hypergraph sizes and degrees need to be determined.\footnote{Bounds for size and degree of oriented hypergraphs are provided in Lemma \ref{cor:size_deg_bounds}.}  This, in turn, requires determining the maximum and minimum number of reactions a given set $V$ of $n$ substances may hold.\footnote{See Lemma \ref{Remark3.5}.}  We call the \emph{complete oriented hypergraph} over $V$ the oriented hypergraph holding the maximum number of reactions over the $n$ substances gathered in $V$.

Given that the adjacency matrix ${\bf M}$ can be arranged according to the size of its hypervertices, it can be conveniently written as
\begin{align}
{\bf M} = \left( \begin{array}{cccc}
{\bf M}_{1,1} & {\bf M}_{1,2} & \cdots & {\bf M}_{1,n-1} \\
{\bf M}_{2,1} & {\bf M}_{2,2} & \cdots & {\bf M}_{2,n-1} \\
\vdots & \vdots & \ddots & \vdots\\
{\bf M}_{n-1,1} & {\bf M}_{n-1,2} & \cdots & {\bf M}_{n-1,n-1} 
\end{array}\right),
\end{align} 
where ${\bf M}_{i,j}$ indicates the block of ${\bf M}$ containing information on the relationship between hypervertices of size $i$ and of size $j$.
\begin{Lemma}\label{rmk:null_blocks}
The blocks ${\bf M}_{i,j}$ with $i+j>n$ are null blocks.
\begin{proof}
As reactions with size $i+j>n$ necessarily have a common substance, out of the $n$ substances, then those reactions gathered in the blocks ${\bf M}_{i,j}$ with $i+j>n$ are impossible.  Therefore, for blocks ${\bf M}_{i,j}$ with $i+j>n$, their ${\bf M}$-entries are 0-entries.
\end{proof}
\end{Lemma}

The effect of Lemma \ref{rmk:null_blocks} upon the possible number of reactions of a given chemical space is enormous, as it shows that the disjoint condition for hypervertices belonging to reactions reduces to a large extent the actual possibilities for exploring new chemical combinations.  Recall that these reactions correspond to red 0-entries in the matrix ${\bf M}$ shown in Table \ref{AMatrixN4}, where it is seen the effect upon a chemical space of only four substances.  Figure \ref{fig:matrices_impossible} shows the proportion of impossible reactions for larger spaces.  Equation \ref{eq:impossible_rxns} below quantifies how rapid impossible reactions grow as a function of $n$.

\begin{figure}[ht!]
\centering
\includegraphics[width=\linewidth]{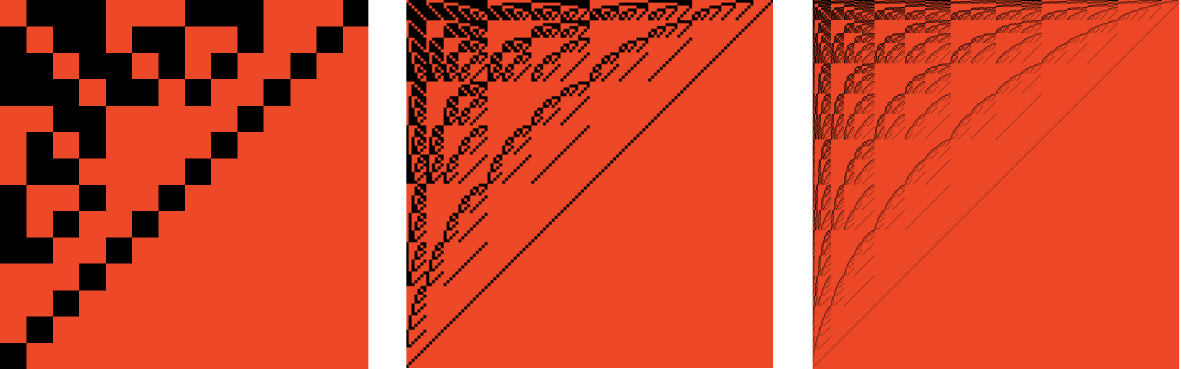}
\caption{Amount of possible and impossible reactions. Visual depiction of adjacency matrices ${\bf M}$ for chemical spaces of a) $n=4$, b) $n=7$ and c) $n=10$ substances (vertices). Possible reactions (black entries) correspond to oriented hyperedges, where the related hypervertices (sets of substances) are disjoint. Impossible reactions (red entries) are the oriented hyperedges relating non-disjoint sets of substances. }
\label{fig:matrices_impossible}
\end{figure}

In order to determine the number of reactions of the complete oriented hypergraph, we analyse the number of reactions in each block ${\bf M}_{i,j}$ of the adjacency matrix.

\begin{Lemma} \label{Remark3.1}
The \emph{number of oriented hyperedges for the block matrix ${\bf M}_{i,j}$} in a complete oriented hypergraph is given by
\begin{align}
u_{i,j}(n) = \left\{ \begin{array}{cc} {\cal C}^n_i {\cal C}^{n - i}_j & i \neq j\\ \frac{1}{2} \, {\cal C}^n_i {\cal C}^{n - i}_i & i = j \end{array}\right.
\end{align}
where ${\cal C}^n_i$ is given in Equation \ref{eq:B_i}.
\begin{proof}
If $ i \neq  j$, the number of oriented hyperedges $u_{i,j}(n)$ corresponds to the number of 1-entries in the block matrix ${\bf M}_{i,j}$, that is $u_{i,j}(n) = \sum_{i,j} {\bf M}_{i,j}$.  Moreover, as ${\bf M}$ is symmetric, $u_{i,j}(n) = u_{j,i}(n)$.  The same symmetry argument indicates that if $i = j$, the number of oriented hyperedges is half the number of 1-entries for the block matrix ${\bf M}_{i,j}$.  As noted, it is important to calculate the number of 1-entries for a block matrix ${\bf M}_{i,j}$.  Let $[ {\bf M}_{i,j} ]$ be the matrix components of the block matrix ${\bf M}_{i,j}$. According to Equation \ref{eq:B_i}, the maximum number of rows of this block is ${\cal C}^n_i$.  Given that vertices are indistinguishable, the number of 1-entries is the same for each row. Hence, the number of 1-entries for ${\bf M}_{i,j}$ is given by ${\cal C}^n_i$ multiplied by the number of 1-entries per row.  As every 1-entry represents a chemical reaction, therefore, the intersection between the hypervertex of the row $i$ and the hypervertex of the column $j$ must be empty. Thus, the number of vertices available for the hypervertex of the column $j$ is given by $n - i$.  Thus, the number of non-zero components on each column is given by ${\cal C}^{n - i}_j$, which is the number of combinations with $j$ vertices out of $n - i$ vertices.  By replacing this result in the previous relations, we obtain
\begin{align}
u_{i,j}(n) &= {\cal C}^n_i \, {\cal C}^{ n - i}_{j} , \nonumber
\end{align}
\noindent which is the total number of non-zero components for the block matrix ${\bf M}_{i,j}$ when $i \neq j$.  If $i = j$, we have to consider half the amount of 1-entries, that is to say
\begin{align}
u_{i,i}(n) &= \frac{1}{2} {\cal C}^n_i \, {\cal C}^{n - i}_{i}. \nonumber
\end{align}
\end{proof}
\end{Lemma}

Knowing the number of reactions in every block of ${\bf M}$, we can determine the number of reactions (oriented hyperedges) of a given size in the complete oriented hypergraph.
\begin{Lemma} \label{Remark3.2}
The \emph{number of oriented hyperedges of size $s$ in a complete oriented hypergraph} is given by
\begin{align} \label{eq:Number_rxns_complete_hypergraph}
u_s(n) = (2^{s-1} - 1){\cal C}^n_s
\end{align}
\begin{proof}
Oriented hyperedges with size $s$ are located within block matrices ${\bf M}_{i,j}$ such that $i+j=s$.  Hence, the  block matrices satisfying this condition are $\{ {\bf M}_{1,s-1}, {\bf M}_{2,s-2}, \dots, {\bf M}_{s-1,1} \}$.  The number of oriented hyperedges of each of these block matrices is  given by the Lemma \ref{Remark3.1}.  Therefore, the number of oriented hyperedges with size $s$ is given by
\begin{align}
u_s(n) &= \frac{1}{2}\sum^{s-1}_{i=1} n \, u_{i,s-i} = \frac{1}{2}\sum^{s-1}_{i=1} {\cal C}^n_i {\cal C}^{n - i }_{s-i} = \frac{1}{2}{\cal C}^n_s \sum^{s-1}_{i=1}  {\cal C}^{s}_{i}
 = \left(2^{s-1} - 1 \right) {\cal C}^n_s \nonumber 
\end{align}
\end{proof}
\end{Lemma}

Thus, for the chemical space of the adjacency matrix shown in Table \ref{AMatrixN4}, $u_2(4)=6$, $u_3(4)=12$ and $u_4(4)=7$.

Now, we can determine the number of reactions in which a substance can participate in a complete oriented hypergraph.

\begin{Lemma}  \label{Remark3.3}
Given a complete oriented hypergraph, the \emph{number of oriented hyperedges in which a vertex participates in the block matrix ${\bf M}_{i,j}$} is
\begin{align}
u_{i,j}(n) = \left\{ \begin{array}{cc} \frac{(i +j)}{n} {\cal C}^{n}_{i+j}\, {\cal C}^{i + j}_{j} & i \neq j \\ \frac{i}{n} {\cal C}^{n}_{2i}\, {\cal C}^{2i}_{i} & i = j \end{array} \right. \label{eqnRemark3.3}
\end{align}
\end{Lemma}

\begin{proof}
Let us consider an arbitrary block matrix ${\bf M}_{i,j}$ and an arbitrary vertex $v_1$. This block matrix can be split into two blocks, one in which the vertex $v_1$ appears in the hypervertex describing the rows of the block and the other block is when the substance appears in the column hypervertex (the intersection is obviously excluded). The number of (row) hypervertices in which the substance $v_1$ appears is given by ${\cal C}^{n-1}_{i-1} $, which is the number of combinations with $i - 1$ vertices out of $n - 1$ vertices. On the other hand, for the same block, the number of (columns) hypervertices in which $v_1$ is not present is given by $ {\cal C}^{n-i}_{j}$, which is the number of combinations with $j$ vertices out of $n - i$ vertices.  Therefore, the total number of oriented hyperedges in which $v_1$ appears in the hypervertex $b_i$ is given by ${\cal C}^{n-1}_{i-1} {\cal C}^{n-i}_{j}$.

Let us now consider the second block within the same block matrix ${\bf M}_{i,j}$. In this case the number of (column) hypervertices in which $v_1$ is contained is given by $ {\cal C}^{n-1}_{j-1} $, which, similarly to the previous case, is the number of combinations with $n - 1$ substances out of $j -1 $ vertices. On the other hand, and still in the second block, the number of (row) hypervertices in which $v_1$ is not present is given by $ {\cal C}^{n-j}_{i} $, which corresponds to the number of combinations of $n-j$ substances out of $i$ substances. Therefore, the total number of oriented hyperedges in which $v_1$ appears in the hyperedge $b_j$ is given by ${\cal C}^{n-j}_{i} {\cal C}^{n-1}_{j-1}$.

Combining these results, we have that the number of oriented hyperedges in which $v_1$ can appear in the block matrix $M_{i,j}$, when $i \neq j$, is given by
\begin{align}
u_{i,j}(n) = {\cal C}^{n-1}_{i-1} {\cal C}^{n-i}_{j}  + {\cal C}^{n-j}_{i} {\cal C}^{n-1}_{j-1} = \frac{(i+j)}{n} \, {\cal C}^{n}_{i+j} \, {\cal C}^{i+j}_{j} \nonumber
\end{align}
\noindent and when $i = j$ we have half the number of oriented hyperedges, that is 
\begin{align}
u_{i,i}(n) = \frac{i}{n} \, {\cal C}^{n}_{2i} \, {\cal C}^{2i}_{i} \nonumber
\end{align}
\end{proof}

The above remarks allow for determining the number of reactions in the complete oriented hypergraph in which a substance can participate (Lemma \ref{Remark3.4}), as well as the number of reactions of a complete oriented hypergraph  (Lemma \ref{Remark3.5}).

\begin{Lemma} \label{Remark3.4}
The \emph{number of oriented hyperedges in which an arbitrary  vertex can belong in a complete oriented hypergraph} is given by
\begin{align}
u(n) =  3^{n-1} - 2^{n-1} . \label{eqnRemark3.4}
\end{align}
\end{Lemma}

\begin{proof}
By considering the result of Lemma \ref{Remark3.1}, we obtain
\begin{align}
u(n) &= \frac{1}{2} \sum^{n-1}_{i=1} \, \sum^{n-i}_{j=1} u_{i,j} =  \frac{1}{2} \sum^{n-1}_{i=1} \, \sum^{n-i}_{j=1} \frac{(i+j)}{n} {\cal C}^n_{i} {\cal C}^{n-i}_j \nonumber \\
&=  \frac{1}{2 n} \left[ \sum^{n-1}_{i=1} \, i \, {\cal C}^n_{i} \sum^{n-i}_{j=1}  {\cal C}^{n-i}_j + \sum^{n-1}_{i=1} \, {\cal C}^n_{i} \sum^{n-i}_{j=1} \, j \, {\cal C}^{n-i}_j \right] \nonumber \\
&=  \frac{1}{2 n} \left\{ \sum^{n-1}_{i=1} \, i \, {\cal C}^n_{i} \left[ 2^{n-i} - 1 \right] + \sum^{n-1}_{i=1} \, {\cal C}^n_{i} \left[ (n - i) 2^{n -i -1} \right] \right\} \nonumber \\
&=  \frac{1}{2 n} \sum^{n-1}_{i=1} {\cal C}^n_i \left\{ \frac{i}{2}  \, 2^{n-i} - i  + \frac{n}{2}  \, 2^{n -i}  \right\} = 3^{n-1} - 2^{n-1}  \nonumber 
\end{align}
\end{proof}

From the chemical perspective, this implies that a substance can participate at most in $u(n)$ chemical reactions, that is to say, the maximum degree of a substance is $u(n)$. A question opened by Lemma \ref{prop:size_equal_degree} was about the bounds for the size and degree of an oriented hypergraph. Lemma \ref{Remark3.4} provides the information to determine them.

\begin{Lemma}\label{cor:size_deg_bounds}
    The size $s(G)$ and degree $d(G)$ of an oriented hypergraph is bounded by $0\leq x(G) \leq n(3^{n-1}-2^{n-1})$, where $x(G)$ stands for either $s(G)$ or $d(G)$.
\end{Lemma}

\begin{proof}
Minimum size and minimum degree of an oriented hypergraph $G$ are reached for the case of a hyperedge-less oriented hypergraph.  Therefore, $\min s(G)$ and $\min d(G)=0$.  The maximum value of these parameters is reached for the complete hypergraph.  As $\max d(G)$ corresponds to the sum of the degree of each vertex in the complete oriented hypergraph, this amounts to add the number of oriented hyperedges in which each vertex in the complete oriented hypergraph belongs.  As $3^{n-1}-2^{n-1}$ (Lemma \ref{Remark3.4}) is the number of oriented hyperedges in which a vertex can belong in the complete oriented hypergraph, summing over all vertices yields $\max d(G)=n(3^{n-1}-2^{n-1})$.  As $d(G)=s(G)$ (Lemma \ref{prop:size_equal_degree}), then, $\max s(G)=n(3^{n-1}-2^{n-1})$
\end{proof}

Thus, for the toy chemical space $G$ depicted in Figure \ref{fig:toy_model}, $s(G)=d(G) \in [0,76]$.  As we found that these figures are equal to 12 for that chemical space, it is therefore observed how far the toy chemical space is from being a complete oriented hypergraph, with $s(G)=d(G)=76$, and how close it is to be a hyperedge-less oriented hypergraph, with $s(G)=d(G)=0$.

Lemma \ref{Remark3.4} also allows for determining the number of reactions housed by a complete oriented hypergraph.
\begin{Lemma} \label{Remark3.5}
The \emph{number of oriented hyperedges for a complete oriented hypergraph} is \cite{Restrepo2022}
\begin{align}
u_r(n) = \frac{1}{2} ( 3^n - 2^{n+1} + 1 ). \label{eqnRemark3.5}
\end{align}
\end{Lemma}

\begin{proof}
By the Lemma \ref{Remark3.1}, it follows that
\begin{align}
u_r(n) &= \frac{1}{2} \sum^{n-1}_{i=1} \sum^{n-i}_{j=1} u_{i,j}(n) = \frac{1}{2} \sum^{n-1}_{i=1} \sum^{n-i}_{j=1} {\cal C}^n_i \, {\cal C}^{n-i}_j = \frac{1}{2} \sum^{n-1}_{i=1} {\cal C}^n_i  \sum^{n-i}_{j=1} {\cal C}^{n-i}_j \nonumber \\
& = \frac{1}{2} \sum^{n-1}_{i=1} {\cal C}^n_i \left[ 2^{n-i} - 1 \right] = \frac{1}{2} \sum^{n-1}_{i=1} {\cal C}^n_i 2^{n-i} - \frac{1}{2} \sum^{n-1}_{i=1} {\cal C}^n_i \nonumber \\
& = \frac{1}{2} \left[ 3^n - 2^n - 1 \right] - \frac{1}{2} \left[ 2^n - 2 \right] = \frac{1}{2} ( 3^n - 2^{n+1} + 1 ) \nonumber
\end{align}
\end{proof}

This indicates that for the toy-example of Figure \ref{fig:toy_model}, the reactions indicated in the chemical space of four substances are only four reactions out of the 25 possible ones, which correspond to the upper or lower triangles of 1-entries above the main diagonal in the adjacency matrix shown in Table \ref{AMatrixN4} plus the 0-entries in black font.  They correspond to the black entries of the upper or lower triangles above the main diagonal in the left matrix of Figure \ref{fig:matrices_impossible}.

Just to have an idea of the speed of growth of $u_r$ regarding $n$, for $n=2$ to 5, $u_r$ takes values 1, 6, 25, and 90.  This growth is given by
\begin{align}
    \frac{du_r}{dn}&=\frac{1}{2}(3^n\ln 3-2^{n+1}\ln 2) 
\end{align}

This quantifies the speed of growth of the possible chemical space as a function of the number of substances of the space.  It constitutes the upper bound of wiring of any chemical space, which sets the stage to contrast this upper bound with the historical record of the chemical space.  This subject is explored in a forthcoming paper.

The number of possible reactions in the complete oriented hypergraph allows for determining the number of impossible reactions because of the disjoint condition of educts and products, which is given by
\begin{align}\label{eq:impossible_rxns}
    z(n)&=(2^n-2)^2-\frac{1}{2}(3^n-2^{n+1}+1) \nonumber \\
    &=\frac{1}{2}(2\cdot 4^{n} -3^n - 6\cdot 2^{n} +7),
\end{align}
which for $n$ ranging from 2 to 5 yields $z=3$, 30, 171 and 810.  In fact,
\begin{align}\label{eq:impossible_rxns}
    \frac{dz}{dn}&=\frac{1}{2}(4^{n} \cdot \ln 16 - 3^n \cdot \ln 3 - 2^{n} \cdot \ln 64 ) ,
\end{align}

which corresponds to the speed of growth of red 0-entries in any adjacency matrix ${\bf M}$.

When comparing the number of possible reactions (Equation \ref{eq:impossible_rxns}) with the number of impossible reactions (Equation \ref{eq:impossible_rxns}), we observe that the former grows much slower than the latter.  This pattern is observed in Figure \ref{fig:matrices_impossible} for different values of $n$.\footnote{A further question is how many of the possible reactions are actually realised by chemists in the chemical space.  This is a subject we address in a forthcoming paper.}

Equipped with the results of this section, we proceed to develop the Erd\H{o}s-Rényi model for oriented hypergraphs.


\section{Erd\H{o}s-Rényi model for oriented hypergraphs}\label{sec:ER_model}

Although the literature on Erd\H{o}s-Rényi-like models for hypergraphs goes back at least 20 years \cite{Spencer2001,Barthelemy2022,Krapivsky2023,DEPANAFIEU2015,KARONSKI2002,Cooper2004,Newman2009,Chodrow2020,PARCZYK2015,Kaminski2019,Dewar2018}, most of those models are devoted to uniform hypergraphs, while a few of them to non-uniform ones \cite{DEPANAFIEU2015,KARONSKI2002,Cooper2004,Newman2009}.  By uniform hypergraphs we mean hypergraphs whose hyperedges have the same cardinality.  Some of those studies explore the statistical and mathematical properties of substructures embedded in random hypergraphs \cite{PARCZYK2015,Kaminski2019,Dewar2018}.  In general, none of those results addresses the particular case of oriented hypergraphs, which is the model we develop in this section.

Given a set of vertices $V$, the random oriented hypergraph corresponds to the realisation, or not, of every possible oriented hyperedge on $V$.  The probability of realisation of these hyperedges is given by $p$.  Like in the Erd\H{o}s-Rényi model for graphs, the random Erd\H{o}s-Rényi oriented hypergraph can be thought as the result of an algorithm that takes every possible couple of disjoint hypervertices on $V$ and decides whether to link them or not.  The decision depends on generating a random number, and if that number happens to be less than a predetermined value, denoted as $p$, then the hypervertices are connected; otherwise, they remain disconnected.

\begin{defin}
An \emph{Erd\H{o}s-Rényi random oriented hypergraph} $G(n,p)$, is an oriented hypergraph with $n$ vertices whose oriented hyperedges result from linking hypervertices with probability $p$.
\end{defin}

This random process leads to particular kinds of probability mass functions for the quantities described in previous sections such as degree and size of oriented hyperedges.  This results from the mathematical consistency of the Erd\H{o}s-Rényi model here presented. The expressions for the probability mass functions are provided in the remaining part of this section. To begin with, the number of reactions $R$ is a binomially distributed random variable, $R \sim \mbox{B}(u_r,p)$, with probability mass function given by

\begin{Prop}
The probability of having a $G(n,p)$ with $R$ oriented hyperedges is
\begin{align}
\mbox{Pr}(R = r) = \left( \begin{array}{c} u_r \\ r \end{array}\right) \, p^r \, \left( 1 - p \right)^{u_r - r},
\end{align}
\noindent which results from realising $r$ reactions and, therefore, of having $u_r - r$ non-realised reactions. The expected value of the number of reactions in $G(n,p)$ is given by
\begin{align} 
\mbox{E}[R] = \sum^{u_r}_{r = 0} r \, \mbox{Pr}(R = r) = p \, u_r, \label{ProbReactions}
\end{align}
\noindent where $u_r$ is given in Equation (\ref{eqnRemark3.5}).
\end{Prop}

This implies that the expected number of reactions in a random oriented hypergraph is proportional to the maximum number of possible reactions.  The actual number is weighted by the probability $p$.

As we discuss in Definitions \ref{def:size_CS} and \ref{def:degree_CS} and in Lemma \ref{prop:size_equal_degree}, size and degree of an oriented hypergraph become important proxies to determine how tight or sparse is a chemical space.  The random model naturally links the tightness of a chemical space with the probability of wiring the associated oriented hypergraph.  Thus, random processes based on high values of $p$ lead to high size and degree oriented hypergraphs, while processes underlying low $p$ figures, necessarily lead to sparse oriented hypergraphs with small sizes and degrees.

The role of $p$ is also central for the probability of observing a given number of reactions of a particular size (Remark \ref{lemma:8.2}).

\begin{remark}\label{lemma:8.2}
The number of oriented hyperedges with size $s$ is a binomially distributed random variable, $R_s \sim \mbox{B}(u_s,p)$, such that its probability mass function is given by
\begin{align}
\mbox{Pr}(R_s = r_s) = \left( \begin{array}{c} u_s \\ r_s \end{array}\right) \, p^{r_s} \, \left( 1 - p \right)^{u_s - r_s} ,
\end{align}
\noindent which results from considering that there are $r_s$ realised reactions of size $s$ and $u_s - r_s$ non-realised reactions with the same size $s$. As a result, the expected value of the number of oriented hyperedges of size $s$ is
\begin{align}\label{eq:exp_size}
\mbox{E}[R_s] = \sum^{u_s}_{r_s=0} r_s \mbox{Pr}(R_s = r_s) = p u_s.
\end{align}
\noindent This leads to determining the probability of having a reaction with size $s$. This probability $P(s)$ is given by the ratio of the expected number of reactions with size $s$ ($\mbox{E}[R_s]$) and the sum of the total number of expected reactions for the different sizes
\begin{align} \label{SizeDistriModel}
P(s) = \frac{ \mbox{E}[R_s] }{\sum^{u_s}_{R_s=2} \mbox{E}[R_s]} = \frac{u_s}{u_r}.
\end{align}
\noindent Hence, $P(s)$ corresponds to the ratio between $u_s$, the number of reactions with size $s$, and the number of possible reactions $u_r$, where $u_s$ and $u_r$ are given in Lemmas \ref{Remark3.2} and \ref{Remark3.5}, respectively. Remarkably, this probability $P(s)$ associated with the size $s$ is $p$-independent in the random model here defined.\footnote{This indicates that, in the case of a phase transition for this model, the probability $P(s)$ cannot be altered by the criticality.}
\end{remark}
Finally, another implication of the random model is that the vertex degree is also a random variable as stated in the following Remark.
\begin{remark} \label{lemma:8.3}
The vertex degree is a binomially distributed discrete random variable, $D \sim \mbox{B}(u_n,p)$, with probability mass function of the form
\begin{align} \label{ProbSDegree}
\mbox{Pr}(D = d) = \left( \begin{array}{c} u_n \\ d \end{array}\right) \, p^d \, \left( 1 - p \right)^{u_n - d},
\end{align}
\noindent which again, results from having $d$ realised reactions for an arbitrary substance out of the $u_n$ reactions in which the substance can participate, and $u_n - d$ non-realised reactions. Therefore, the expected value of the vertex degree is
\begin{align} \label{EVD}
\mbox{E}[D] = \sum^{u_n}_{d=0} d \, \mbox{Pr}(D = d) = p \, u_n,
\end{align}

\noindent where $u_n$ is given in Equation \ref{eqnRemark3.4}.
\end{remark}

Equations \ref{ProbReactions}, \ref{eq:exp_size}, \ref{SizeDistriModel} and \ref{EVD} show that our model is conceptually correct. As a consequence, in any test for randomness of real data modelled as an oriented hypergraph, each of the  probability distributions discussed here should be statistically close to those given in the aforementioned equations, if the data are randomly generated. On the other hand, as the expected number of reactions of the random oriented hypergraph is related to the expected number of reactions a given substance has, we obtain the following expression for the ratio of those two variables
\begin{align}
\frac{ \mbox{E}[R]}{\mbox{E}[D]} = \frac{1}{2} \frac{\left[ 1 - 2 \left( \frac{2}{3} \right)^n + 3^{-n} \right]}{\left[ \frac{1}{3} - \frac{1}{2} \left( \frac{2}{3} \right)^n\right]},
\end{align}
which, for a large number of substances leads to
\begin{align}
\lim_{n \rightarrow \infty} \frac{\mbox{E}[R]}{\mbox{E}[D]} = \frac{3}{2}. \label{RandD}
\end{align}

That is, if a chemical space is randomly wired, the number of reactions it houses is $3/2$ the number of reactions in which any substance of the space participates.  Therefore, the ratio $ \mbox{E}[R] / \mbox{E}[D] $ can be used to test whether a given chemical space is close to randomness or not. This result clearly contrast with its Erd\H{o}s-Rényi analog for simple graphs which takes the form
\begin{align}
\frac{\mbox{E}[R]}{\mbox{E}[D]} = \frac{n}{2},
\end{align}
and which grows linearly with the number of substances. In the case of the chemical space oriented hypergraph, the factor $3/2$ is actually an upper bound to the ratio $\mbox{E}[R]/\mbox{E}[D]$.

Aiming at having more insight on the effect of $p$ upon the relation between the number of substances ($n$) and the expected number of reactions ($\mbox{E}[R]$), we explore different forms $p$ might take. They range from the case of a constant number of reactions, independent of the number of substances ($\mbox{E}[R] \sim k$); or from a simple linear relation $\mbox{E}[R] \sim k n$; to more complex relations in which the expected number of reactions varies according to a power of the number of substances ($\mbox{E}[R] \sim n^\alpha$); or even that the number of reactions grows exponentially with the number of substances ($\mbox{E}[R] \sim b^n$). To do so, we analyse some cases of chemical spaces for which different values of $n$ and $p$ are considered.

From Equation \ref{ProbReactions} we know that for large values of $n$, $\ln \mbox{E}[R]$ takes the form\footnote{It is known that for the actual chemical space $n \sim 10^6$ \cite{Llanos2019}.}
\begin{align}
\ln \mbox{E}[R] \sim \alpha \ln{n} + n \ln\left( \frac{3}{\beta}\right) ,\label{Limit1}
\end{align}
where the above discussed chemical spaces are generalised by considering a probability given as 
\begin{align}
p=n^\alpha/\beta^n. \label{Probab}
\end{align}
With $\beta=3$, Equation \ref{Limit1} becomes
\begin{align}
\ln \mbox{E}[R] \sim \alpha \ln{n} ,\label{Limit2}
\end{align}
which is a linear relation in a log-log scale with $\alpha$ encoding the slope of the linear trend.  When $\alpha = 0$, $\mbox{E}[R] \sim 1/2$, in which case $p=1/3^n$ and no matter how large the number of substances in the chemical space is, the number of reactions reported by chemists is always a fixed number. $\mbox{E}[R] \sim n$ is obtained with $\alpha=1$, where $p=n/3^n$. This linear relation between $\mbox{E}[R]$ and $n$ indicates that chemists manage wiring the space in a manner that is proportional to the available substances. $\mbox{E}[R] \sim n^\alpha$ is obtained with $p= n^\alpha/3^n$.  If $\alpha > 1 $, the greater the value of $\alpha$, the more reactions are discovered. The scenario with $\alpha < 0 $ yields a sparse chemical space with a decreasing power law relation in which the more substances, the less reactions are discovered. These behaviours are shown in Figure \ref{fig:exp_rxns}a for $\alpha  = -1, 0, 1$ and $2$. 

In turn, $\mbox{E}[R] \sim b^n$, actually $\mbox{E}[R] \sim 3^n$ is reached with $\beta=1$, in which case $p=n^\alpha$ and the leading order of Equation \ref{Limit1} gives
\begin{align}
\ln \mbox{E}[R] \sim n \ln\left( 3 \right) .\label{Limit3}
\end{align}
Hence, the log-plot of the $ \mbox{E}[R]$ as a function of $n$ depicts a constant slope for different values of $\alpha$, as can be seen in Figure \ref{fig:exp_rxns}b. This latter result follows from the fact that for large values of $n$ and fixed and small values of $\alpha$, the first term is negligible (note that $\alpha \leq 0$ to secure that $0\leq p \leq 1$).

These results, besides their importance for the analysis of chemical spaces, pave the way for the exploration of phase transitions in Erd\H{o}s-Rényi-like oriented hypergraphs.  In this respect, although chemical spaces with $\mbox{E}[R] \ll 1$ lack chemical meaning,\footnote{Which occur for low values of $n$ in Figure \ref{fig:exp_rxns}.} they turn interesting to analyse the aforementioned phase transitions. 

\begin{figure}[ht!]
  \centering
  \includegraphics[width=0.7\linewidth]{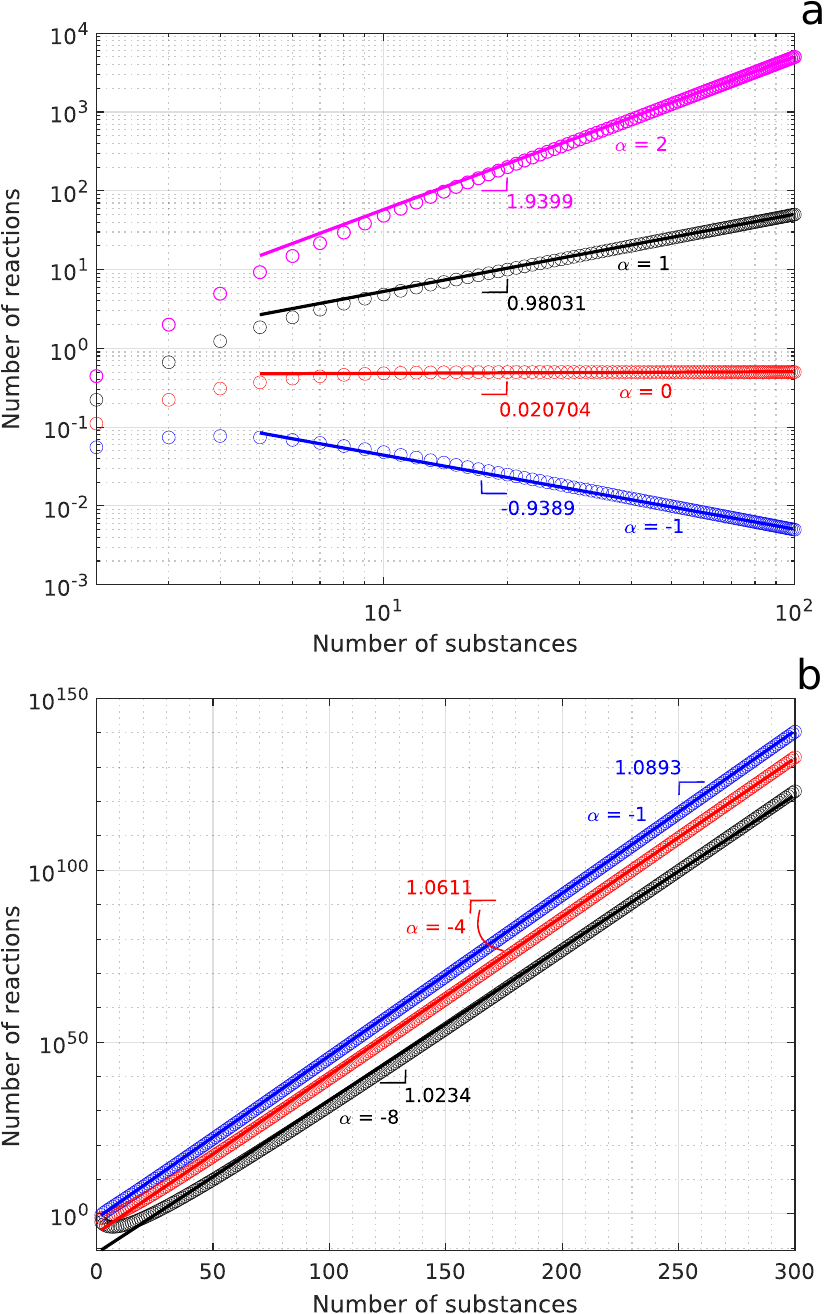}
  \caption{Effects of the probability of triggering chemical reactions upon the expected number of reactions of randomly wired chemical spaces.  Probability is expressed as $p=n^\alpha/\beta^n$ and the plots show how the expected number of reactions $\mbox{E}[R]$ varies with the selection of $\alpha$ and $\beta$. In a) $\beta = 3$ and $\alpha$ takes different values, which show the decreasing power law ($\alpha = -1$), the constant ($\alpha=0$), linear ($\alpha=1$) and quadratic ($\alpha=2$) growth of $\mbox{E}[R]$ for large values of the number of substances $n$. In all these chemical spaces, where $\beta=3$, $\alpha \leq (n\ln{3})/(\ln{n})$ to warranty that $0 \leq p \leq 1$. In b) $\beta=1$ and $\alpha \leq 0$ to secure that $0 \leq p \leq 1$. These plots correspond to exponential-like growths of $\mbox{E}[R]$ for large values of $n$, where the slope of the linear fit tend to $\ln 3 \approx 1.099$. Plots in a) and b) were obtained for different values of $n$ in $\mbox{E}[R] = \frac{n^\alpha}{2\beta^n}(3^n-2^{n+1}+1)$.}
        \label{fig:exp_rxns}
\end{figure}

The probability of triggering a reaction not only affects the number of reactions but also the size of those reactions in the chemical space.  Therefore, we explore how the different values of $p$, given in Equation \ref{Probab}, affect the number of reactions of different sizes. From Equation \ref{eq:exp_size} we know that the expected number of reactions of size $s$ ($\mbox{E}[R_s]$) is given by $pu_s$.  That is, $\mbox{E}[R_s]$ results from the probability of realising or not reactions of size $s$ in the complete oriented hypergraph.  By operating on the expressions for $u_s$ (Lemma \ref{Remark3.2}), we found, for large values of $n$ and with $\beta = 1$, that
\begin{align}
\ln \mbox{E}[R_s] \sim (s + \alpha)\ln{n}, \label{SizeLimit1}
\end{align}
where we used the binomial coefficient approximation for large values of $n$ and fixed (small) values of $s$ together with the Stirling approximation \cite{Asymptopia}. This expression is similar to Equation \ref{Limit2} in what it shows three regimes attending to the value of slope $s + \alpha$. 

Recalling that $\alpha < 0$ to guarantee that $0 \leq p \leq 1$, this result indicates that in a random chemical space where the number of reactions grows exponentially with the number of substances (Equation \ref{Limit3}), the number of reactions with size either $s$ drops, remains constant or increases, depending on whether the slope $s + \alpha$ is negative, null or positive, respectively. The general expression for large values of $n$ is given by $\mbox{E}[R_s] = n^{s -|\alpha|}$. For example, if $\alpha = -2$, rearrangement reactions ($s=2$) remain constant with the number of available substances $n$, while reactions with size $s > 2$ follow a power law $\mbox{E}[R_s] \sim n^{s - 2}$. Given that the smallest value of $s$ is $2$, for $\alpha = -2$ there is no possibility of a negative slope. However, for a different value, for example $\alpha = -4$, reactions with size $s < 4$ drop following a power law and give rise to a sparse population of reactions with those sizes. In Figure \ref{fig:sizes}a we show an example for $\alpha = -2$, $\beta = 1$ and $s = 2, 3, 4$ and $5$.

\begin{figure}[ht!]
\centering
\includegraphics[width=0.7\textwidth]{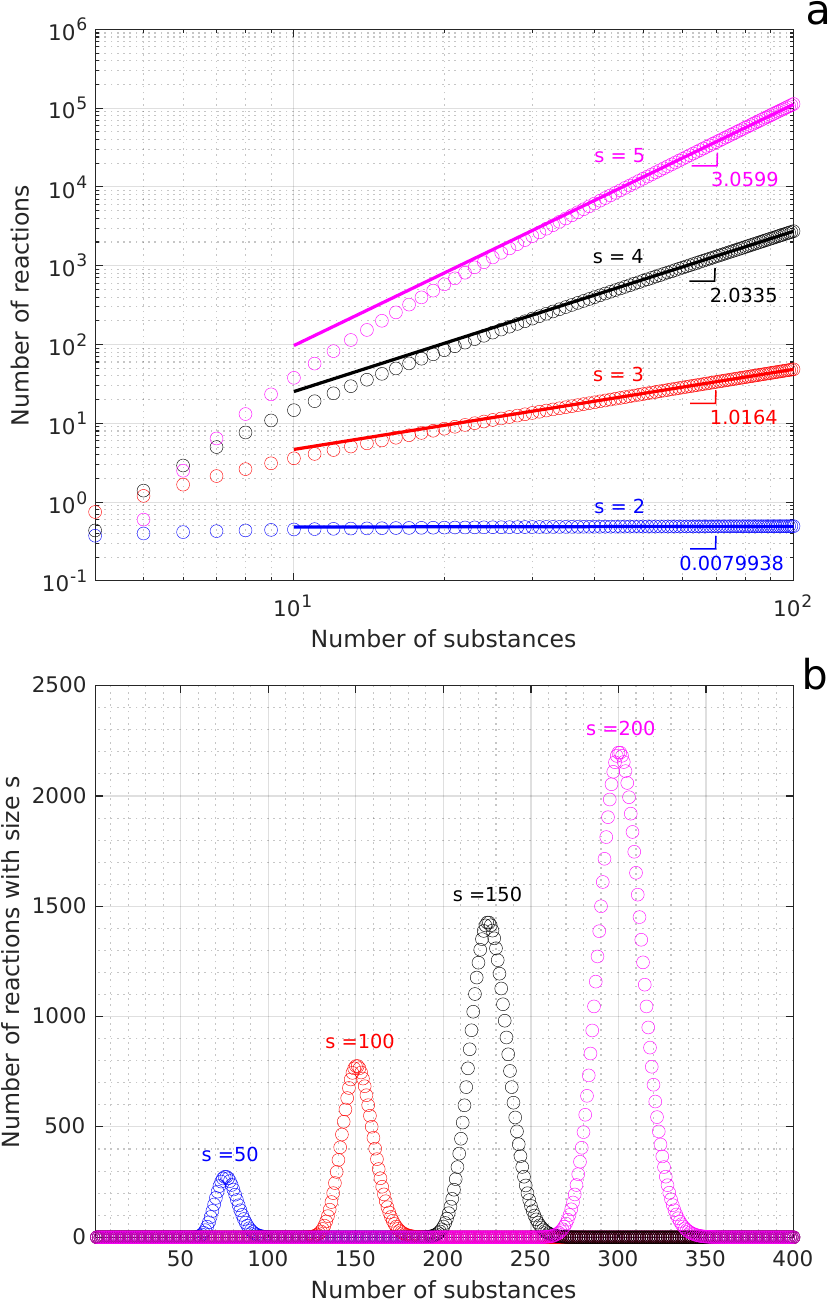}
\caption{Effects of the probability of triggering reactions upon the expected number of reactions of different sizes in a randomly wired chemical space.  Probability is given by $p=n^\alpha /\beta^n$.  a) Behaviour at $\alpha = -2$ and $\beta = 1$, which corresponds to a chemical space whose number of reactions exponentially expand with the number of substances only in case $s > |\alpha|$.  b) Distribution of number of reactions of size $s = 50, 100, 150$ and $200$ for chemical spaces with $\alpha = 2$ and $\beta=3$, corresponding to spaces whose number of reactions grows at a power law of the number of substances. Maximum values are given at $n_{max} = 76, 151, 226$ and $301$ respectively.}
\label{fig:sizes}
\end{figure}

As $\beta$ may take also the value of 3 for chemical spaces with either linear or power-law growth of the number of reactions, for these cases the leading order of $\ln \mbox{E}[R_s]$ takes the form
\begin{align}
\ln \mbox{E}[R_s] \sim -n\ln{3}, \label{SizeLimit2}
\end{align}
which shows that the asymptotic behaviour of $\mbox{E}[R_s]$ for large values of $n$ is a decreasing exponential $\mbox{E}[R_s] \sim 1/3^{n}$ in terms of the number of substances $n$ and for any size $s$, $s$ being much smaller than $n$.\footnote{It is known that $s \ll n$ for actual chemical spaces \cite{Llanos2019}.} However, the number of reactions with a fixed size $s$, $\mbox{E}[R_s]$, reaches a maximum value at a number of substances $n_{max}$ given by the solution to the equation
\begin{align}
\frac{\alpha}{n_{max}} + \ln\left[ \frac{n_{max}}{3(n_{max} -s)}\right] = 0
\end{align}
where again, we used the Stirling approximation for large values of $n$ \cite{Asymptopia}. This implies that in a random linear or power-law wiring of the chemical space, the number of reactions of size $s$ reaches its maximum population with spaces of $n=n_{max}$ substances (Figure \ref{fig:sizes}b). This result shows another facet of randomly wired chemical spaces.  In particular, that large randomly wired spaces are mainly populated by reactions of large sizes, where reactions of small size only represent a small population of the bulk of reactions.  This suggests that actual chemical spaces, mainly populated by reactions of small size \cite{Llanos2019}, are indeed far from a random wiring.

So far, we have compared the values of $\mbox{E}[R_s]$ for different realisations of a random chemical space. To explore how the number of reactions with different sizes $s$ relates each other within the same chemical space, we now fix $n$ and note that, for $\beta = 1$ the probability of wiring reactions, $p = 1/n^{\alpha}$ is much larger than $p = n^{\alpha}/3^n$, the wiring probability with $\beta  = 3$ with $\alpha$ fixed in both cases. Consequently, the number of reactions $\mbox{E}[R]$ is much larger in the first case than in the second, see Figure \ref{fig:exp_rxns}. This implies that the number of reactions with a given size $s$ follows the same trend, that is to say, the number of reactions with size $s$, $\mbox{E}[R_s]$, is larger when the probability is given by $p = 1/n^{\alpha}$ compared with what would be the number of reactions of the same size for a lower probability $p = n^{\alpha}/3^n$, see Figure \ref{fig:sizes2}. Additionally, as long as $p$ is considered to be $s-$independent, for a given number of substances $n$, $\mbox{E}[R_s]$ reaches a maximum at a $p$-independent size value $s_{max}$. The value of the most populated size $s_{max}$ is the solution to the equation (see Figure \ref{fig:sizes2}):
\begin{align}
\frac{2^{s_{max} - 1} \ln 2}{2^{s_{max} -1} -1} + \ln\left( \frac{n - s_{max}}{s_{max}}\right) = 0,
\end{align}
where we used the Stirling approximation for large values of $n$ \cite{Asymptopia}.

\begin{figure}[ht!]
\centering
\includegraphics[width=0.8\textwidth]{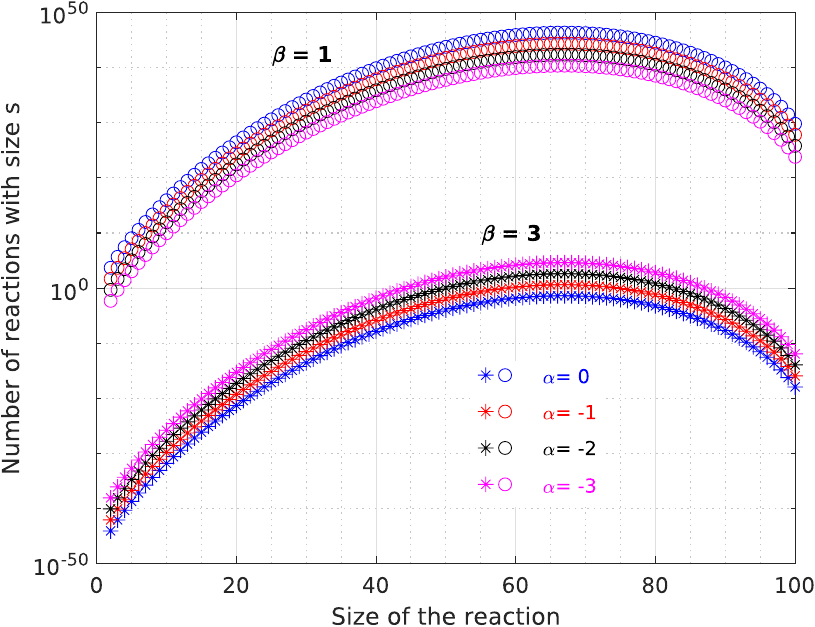}
\caption{Effects of the probability of triggering reactions upon the expected number of reactions of different sizes in randomly wired chemical spaces with a fixed number of substances, $n = 100$ and wiring probability given by $p=n^\alpha /\beta^n$. Behaviour for different values of $\alpha $ and $\beta$. For $\beta = 3$, $\alpha = 0, 1, 2$ and $3$, while for $\beta = 1$, $\alpha = -3, -2, -1$ and $0$. The maximum values of $\mbox{E}[R_s]$ are given at $s_{max} = 67$.}
\label{fig:sizes2}
\end{figure}



\section{Conclusion and outlook}
We developed the Erd\H{o}s-Rényi model for oriented hypergraphs with a fixed number $n$ of vertices.  Oriented hypergraphs result from the binary relations between sets of arbitrary size of vertices (hypervertices).  In particular, we considered oriented hypergraphs where the related hypervertices are disjoint.  This follows from our aim of modelling all possible chemical reactions, that is catalysed and non-catalysed reactions, with the former including autocatalytic reactions.

Central for the Erd\H{o}s-Rényi model is the complete oriented hypergraph, as the model randomly realises oriented hyperedges of the complete oriented hypergraph.  This realisation is mediated by a probability $p$.  We analysed different functional forms of $p$, which depend on $n$, and that allow for constant, linear, power law and exponential behaviours of the number of oriented hyperedges as a function of $n$.  These forms of $p$, as well as the trends and their effects upon the number of oriented hyperedges constitute an approach to determine whether a given oriented hypergraph follows the patterns of a random wired oriented hypergraph.

The application motivating this study is the chemical space, which we model as an oriented hypergraph, where vertices correspond to substances and oriented hyperedges to chemical reactions.  Two main reasons turn oriented hypergraphs as a suitable model for chemical spaces, or in general, sets of chemical reactions.  First, hypervertices, which are the relata of the oriented hypergraph model, gather together the substances involved in a chemical reaction.  Therefore, hypervertices turn instrumental to distinguish substances involved in a chemical transformation from those which do not.  Second, oriented hyperedges, that is binary collections of hypervertices, encode a fundamental aspect of chemical reactions, namely the distinction between educts and products.

The Erd\H{o}s-Rényi model here formulated is central for chemical studies since the availability of chemical information in electronic form, spanning centuries of chemical activity, as well as the ever growing computational power, have made possible to study the evolution of the entire chemical space \cite{Llanos2019,Restrepo2022}.  This has led to  results on the growth rate of the space, as well as to analyses of its diversity  at the compositional and molecular structural levels, as well as at the reaction type level \cite{Llanos2019,Restrepo2022,Fialkowski2005,Grzybowski2021,Lipkus2019}.  Despite the importance of these studies, an open question is whether the chemical space has a particular structure and how far, or close, such a structure lies from its random counterpart \cite{JostRestrepo2022}.  This boils down to the question whether chemical space has been randomly wired in any moment of its expansion, which poses further questions on the nature of chemistry as a discipline \cite{JostRestrepo2022}.  At any rate, whether chemical space is far or close to randomness requires quantifying the distance between the actual chemical space at a particular time and its random case.  The Erd\H{o}s-Rényi model, therefore, turns central for such a quantification, which is the subject of a  forthcoming publication.  In this respect, by analysing the number of reactions involving a certain number of substances in a randomly wired chemical space, we found evidence that the actual chemical space seems to depart from a random wiring, as random spaces with similar number of substances of those in the actual chemical space are mainly populated by reactions involving by far more than a handful of substances.  This latter is the typical situation of actual chemical reactions.  The same approach of exploring the distribution of reactions involving certain number of substances over time turns central to analyse whether there are subspaces of the chemical space of rather small number of substances that are close to a random wiring.  A similar argument may be used to analyse the random wiring of spaces going beyond the traditional two or three educts of chemical reactions \cite{Llanos2019}, which is one of the objectives of one-pot reactions \cite{Hayashi2021}, for instance, which aim at affording a circular and green chemistry.

As recently discussed \cite{Mueller:22a}, oriented hypergraphs become suitable models to encode the molecular reality of reactions if every oriented hyperedge is conservative.  That is, if its associated stoichiometric matrix, which encodes the actual amount of substances involved in the reaction, has a strictly positive reaction invariant.  This mathematical condition implies that a chemical reaction must preserve the number of atoms and, therefore, mass.  The condition imposed by the stoichiometric matrix triggers further questions.  For instance, one might inquire into the frequency with which the random model fulfills the given criterion as a function of the number of oriented hyperedges. Our model can be extended to incorporate such stoichiometric analyses but it implies several changes in the adjacency matrix, which plays a fundamental role in obtaining most of the expressions used here. Such a modification requires further investigation and it is beyond the scope of the current paper.

The Erd\H{o}s-Rényi model is general enough to be applied to non-chemical systems modelled by oriented hypergraphs, which include description of logical and artificial intelligence processes, as well as database modelling \cite{Ausiello2017}.  Particular examples include the case of functional dependencies in relational databases \cite{Ausiello1983} or the analysis of Horn formulae in logic \cite{WILD2017264} and the study of AND-OR graphs \cite{Gallo1998}.

While developing the Erd\H{o}s-Rényi model, we introduced concepts of further interest for the study of oriented hypergraphs such as size and degree of oriented hyperedges, which we extended to the size and degree of the entire oriented hypergraph structure made by the collection of oriented hyperedges.  In chemical terms, we defined the size and degree of any reaction and we extended these concepts to the size and degree of arbitrary chemical spaces.  The size of an oriented hyperedge corresponds to the number of vertices belonging in the hyperedge, while the degree of the oriented hyperedge accounts for the number of oriented hyperedges incident to the vertices of the oriented hyperedge of reference.  In chemical terms, the size of a reaction accounts for the number of substances in the reaction and the degree of the reaction for the number of reactions in which substances of the reaction in question participate.  We showed that the size and degree of an oriented hypergraph are equal.  They indicate whether an oriented hypergraph is sparse or loosely connected, which occurs for structures with low size and degree values, that is close to 0.  In contrast, oriented hypergraphs with high values, close to the upper bound of size and degree ($n(3^{n-1}-2^{n-1})$) turn to be tightly connected structures.

By analysing the complete oriented hypergraph for $n$ vertices, we found that the maximum number of oriented hyperedges incident to any vertex is $3^{n-1}-2^{n-1}$, which in chemical terms amounts to determining the maximum number of reactions for a substance in a chemical space.  This expression evidences the extremely large possibilities for wiring the chemicals space \cite{Restrepo2022}.  This result led to find, as discussed before, that the size and degree of any oriented hypergraph are restricted to the interval [0, $n(3^{n-1}-2^{n-1})$].  The extreme wiring possibilities of a chemical space of $n$ substances are embodied in the maximum number of reactions a chemical space may hold, which turns out to be $\frac{1}{2}(3^n-2^{n+1}+1)$.  This is the number of oriented hyperedges of the complete oriented hypergraph.  The aforementioned result strongly contrast with the $n(n-1)/2$ edges of a graph.  The fact that the maximum number of reactions a chemical space may hold is proportional to $3^n$, when modelled as an oriented hypergraph, contrasts sharply with the just $n^2$ edges of the models of the chemical space based on graphs.  These huge differences between the number of oriented hyperedges and edges are the result of the much richer description of chemical reactions provided by the oriented hypergraph model, where the two sets of substances involved in chemical reactions (educts and products) are an explicit part of the model, whereas in the graph model, they are just disregarded.

The number of reactions in the complete oriented hypergraph allowed for determining the speed of growth of its oriented hyperedges as a function of the number of vertices ($du_r/dn$). We found that $du_r/dn=\frac{1}{2}(3^n\ln{3}-2^{n+1}\ln{2})$.  This result bears important implications for chemistry, as it provides the upper bound for the growth of the number of chemical reactions as a function of the number of available substances in the chemical space.  This allows for determining that the expected number of reactions of a random oriented hypergraph is given by $\frac{p}{2}(3^n-2^{n+1}+1)$ and for deriving similar expressions for the expected number of reactions of size $s$.  These results, in turn, allow to contrast actual chemical spaces with random chemical spaces at different levels of detail, for instance by analysing the actual and expected number of reactions of particular sizes.  An important invariant of random oriented hypergraphs we found is the ratio of the expected number of reactions and the expected degree or size of the chemical space. For large values of $n$, we found this ratio is 3/2, which becomes a proxy to determine whether a given chemical space is random of not, according to the Erd\H{o}s-Rényi model here presented.

Despite the richness of the oriented hypergraph for modelling chemical spaces, it is open to improvements.  For instance by introducing direction between the sets of substances involved in reactions, which would clearly distinguish between educts and products. In this case, chemical spaces are modelled as directed hypergraphs.  The Erd\H{o}s-Rényi model here presented requires further refinements to be adjusted for directed hypergraphs.  They involve, for instance, to adjust $p$ to distinguish between $ X \rightarrow Y$ and $Y \rightarrow X$, with $X$ and $Y$ being sets of substances (hypervertices).  We discussed the functional form of $p$ as depending of $n$, but other forms of $p$ are also worth exploring, for instance as a function of reaction size ($s$).  This leads to explore how the wiring of random chemical spaces depends on the amount of substances involved in their chemical transformations, which is a determining factor of the chemical possibilities to trigger actual chemical reactions \cite{child2014molecular}.

Besides the Erd\H{o}s-Rényi model for the high-order structures discussed in this paper, namely hypergraphs, as well as directed and oriented hypergraphs, other models need to be defined upon them, for instance the small-world \cite{Watts1998}, as well as the Barab\'asi-Albert \cite{Albert2002} ones. As well as in the original Erd\H{o}s-Rényi model, were phase transitions and the conditions to attain them were studied, a further avenue of research is the study of those conditions to afford phase transitions in higher-order structures as the ones here presented.

\section{Author contributing}
AG-C conducted the research, AG-C and GR conceptualised the project, GR wrote the paper and AG-C, MBM, PFS, JJ and GR discussed and reviewed the edited document.

\section{Conflicts of Interest Statement}
We have no competing interests.

\section{Funding}
MBM thanks the  support of the Alexander von Humboldt Foundation.

\section{ Acknowledgments}
The authors thank the feedback from Guido Montufar and Humberto Laguna upon early results of this project.

\bibliographystyle{prsa4}
\bibliography{References}

\begin{thebibliography}{10}
\expandafter\ifx\csname urlstyle\endcsname\relax
  \providecommand{\doi}[1]{doi:\discretionary{}{}{}#1}\else
  \providecommand{\doi}{doi:\discretionary{}{}{}\begingroup
  \urlstyle{rm}\Url}\fi

\bibitem{Schmidt2011}
Schmidt G, 2011 \emph{Relational Mathematics}.
\newblock Encyclopedia of Mathematics and its Applications. Cambridge
  University Press.

\bibitem{Diestel2018}
Diestel R, 2018 \emph{Graph Theory}.
\newblock Graduate Texts in Mathematics. Springer Berlin Heidelberg.

\bibitem{ErdosRenyi1959}
Erd\H{o}s P, R\'{e}nyi A, 1959 On Random Graphs I.
\newblock \emph{Publicationes Mathematicae Debrecen} \textbf{6}, 290.

\bibitem{Boccaletti2023}
Boccaletti S, {De Lellis} P, {del Genio} C, Alfaro-Bittner K, Criado R, Jalan
  S, Romance M, 2023 The structure and dynamics of networks with higher order
  interactions.
\newblock \emph{Physics Reports} \textbf{1018}, 1--64.
\newblock (\doi{https://doi.org/10.1016/j.physrep.2023.04.002)}.

\bibitem{Chodrow2020}
Chodrow PS, 2020 Configuration models of random hypergraphs.
\newblock \emph{Journal of Complex Networks} \textbf{8}, 3, cnaa018.
\newblock (\doi{10.1093/comnet/cnaa018)}.

\bibitem{Yixin2021}
Ji Y, Zhang Y, Shi H, Jiao Z, Wang SH, Wang C, 2021 Constructing Dynamic Brain
  Functional Networks via Hyper-Graph Manifold Regularization for Mild
  Cognitive Impairment Classification.
\newblock \emph{Frontiers in Neuroscience} \textbf{15}.
\newblock (\doi{10.3389/fnins.2021.669345)}.

\bibitem{Shi2017}
Gu S, Yang M, Medaglia JD, Gur RC, Gur RE, Satterthwaite TD, Bassett DS, 2017
  Functional hypergraph uncovers novel covariant structures over
  neurodevelopment.
\newblock \emph{Human Brain Mapping} \textbf{38}, 8, 3823--3835.
\newblock (\doi{https://doi.org/10.1002/hbm.23631)}.

\bibitem{Murgas2022}
Murgas KA, Saucan E, Sandhu R, 2022 Hypergraph geometry reflects higher-order
  dynamics in protein interaction networks.
\newblock \emph{Scientific Reports} \textbf{12}, 1, 20879.
\newblock (\doi{10.1038/s41598-022-24584-w)}.

\bibitem{Restrepo2022}
Restrepo G, 2022 Chemical space: Limits{,} evolution and modelling of an object
  bigger than our universal library.
\newblock \emph{Digital Discovery} \textbf{1}, 568--585.
\newblock (\doi{10.1039/D2DD00030J)}.

\bibitem{Leal2021}
Leal W, Restrepo G, Stadler PF, Jost J, 2021 Forman–{R}icci curvature for
  hypergraphs.
\newblock \emph{Advances in Complex Systems} \textbf{24}, 01, 2150003.
\newblock (\doi{10.1142/S021952592150003X)}.

\bibitem{Stadler2015}
Flamm C, Stadler BMR, Stadler PF, 2015 \emph{Generalized Topologies:
  Hypergraphs, Chemical Reactions, and Biological Evolution}, chapter~2, pages
  300--328.
\newblock Bentham-Elsevier.

\bibitem{Stadler2018}
Stadler BMR, Stadler PF, 2018 Reachability, Connectivity, and Proximity in
  Chemical Spaces.
\newblock \emph{MATCH Communications in Mathematical and in Computer Chemistry}
  \textbf{80}, 3, 639--659.

\bibitem{Menezes2021}
Menezes T, Roth C, 2021.
\newblock Semantic Hypergraphs.
\newblock (\doi{https://doi.org/10.48550/arXiv.1908.10784)}.

\bibitem{Zuguo2017}
Chen Z, Li Y, Chen X, Yang C, Gui W, 2017 Semantic Network Based on
  Intuitionistic Fuzzy Directed Hyper-Graphs and Application to Aluminum
  Electrolysis Cell Condition Identification.
\newblock \emph{IEEE Access} \textbf{5}, 20145--20156.
\newblock (\doi{10.1109/ACCESS.2017.2752200)}.

\bibitem{Xie2021}
Xie Z, 2021 A distributed hypergraph model for simulating the evolution of
  large coauthorship networks.
\newblock \emph{Scientometrics} \textbf{126}, 6, 4609--4638.
\newblock (\doi{10.1007/s11192-021-03991-2)}.

\bibitem{Mulas2022}
Mulas R, Horak D, Jost J, 2022 \emph{Graphs, Simplicial Complexes and
  Hypergraphs: Spectral Theory and Topology}.
\newblock Cham: Springer International Publishing.
\newblock (\doi{10.1007/978-3-030-91374-8\_1)}.

\bibitem{Mulas2019}
Jost J, Mulas R, 2019 Hypergraph Laplace operators for chemical reaction
  networks.
\newblock \emph{Advances in Mathematics} \textbf{351}, 870--896.

\bibitem{Eidi2020}
Eidi M, Jost J, 2019 {Ollivier Ricci curvature of directed hypergraphs}.
\newblock \emph{Scientific Reports} \textbf{10}, 1, 12466.

\bibitem{Eidi2020a}
Eidi M, Farzam A, Leal W, Samal A, Jost J, 2020 {Edge-based analysis of
  networks: Curvatures of graphs and hypergraphs}.
\newblock \emph{Theory in Biosciences} \textbf{139}, 4, 337--348.

\bibitem{Mulas2020a}
Mulas R, Kuehn C, Jost J, 2020 Coupled dynamics on hypergraphs: Master
  stability of steady states and synchronization.
\newblock \emph{Physical Review E} \textbf{101}, 062313.
\newblock (\doi{10.1103/PhysRevE.101.062313)}.

\bibitem{Sun2021}
Sun H, Bianconi G, 2021 Higher-order percolation processes on multiplex
  hypergraphs.
\newblock \emph{Phys. Rev. E} \textbf{104}, 034306.
\newblock (\doi{10.1103/PhysRevE.104.034306)}.

\bibitem{Thakur2009}
Thakur M, Tripathi R, 2009 Linear connectivity problems in directed
  hypergraphs.
\newblock \emph{Theoretical Computer Science} \textbf{410}, 27, 2592--2618.
\newblock (\doi{https://doi.org/10.1016/j.tcs.2009.02.038)}.

\bibitem{Ausiello2017}
Ausiello G, Laura L, 2017 Directed hypergraphs: Introduction and fundamental
  algorithms—A survey.
\newblock \emph{Theoretical Computer Science} \textbf{658}, 293--306.
\newblock (\doi{https://doi.org/10.1016/j.tcs.2016.03.016)}.

\bibitem{Llanos2019}
Llanos EJ, Leal W, Luu DH, Jost J, Stadler PF, Restrepo G, 2019 Exploration of
  the chemical space and its three historical regimes.
\newblock \emph{Proceedings of the National Academy of Sciences} \textbf{116},
  26, 12660--12665.

\bibitem{Fialkowski2005}
Fialkowski M, Bishop KJM, Chubukov VA, Campbell CJ, Grzybowski BA, 2005
  Architecture and Evolution of Organic Chemistry.
\newblock \emph{Angewandte Chemie International Edition} \textbf{44}, 44,
  7263--7269.

\bibitem{Grzybowski2006}
Bishop KJM, Klajn R, Grzybowski BA, 2006 The core and most useful molecules in
  organic chemistry.
\newblock \emph{Angewandte Chemie International Edition} \textbf{45}, 32,
  5348--5354.

\bibitem{JostRestrepo2022}
Jost J, Restrepo G, 2022 \emph{The Evolution of Chemical Knowledge: A Formal
  Setting for its Analysis}.
\newblock Wissenschaft und Philosophie -- Science and Philosophy -- Sciences et
  Philosophie. Springer International Publishing.

\bibitem{Mueller:22a}
M{\"u}ller S, Flamm C, Stadler PF, 2022 What makes a reaction network
  ``chemical''?
\newblock \emph{J. Cheminformatics} \textbf{14}, 63.
\newblock (\doi{10.1186/s13321-022-00621-8)}.

\bibitem{Dittrich2007}
Dittrich P, di~Fenizio PS, 2007 Chemical Organisation Theory.
\newblock \emph{Bulletin of Mathematical Biology} \textbf{69}, 4, 1199--1231.
\newblock (\doi{10.1007/s11538-006-9130-8)}.

\bibitem{Schummer1996}
Schummer J, 1996 \emph{Realismus und Chemie: philosophische Untersuchungen der
  Wissenschaft von den Stoffen}.
\newblock Epistemata / Reihe Philosophie: Reihe Philosophie. K{\"o}nigshausen
  \& Neumann.

\bibitem{Atkins2023}
Atkins P, de~Paula J, Keeler J, 2023 \emph{Atkins' Physical Chemistry}.
\newblock Oxford University Press.

\bibitem{lovasz1993}
Lov{\'a}sz L, 1993 \emph{Combinatorial Problems and Exercises}.
\newblock AMS/Chelsea publication. North-Holland Publishing Company.

\bibitem{Spencer2001}
Spencer J, 2001 \emph{The Strange Logic of Random Graphs}.
\newblock Algorithms and Combinatorics. Springer Berlin Heidelberg.

\bibitem{Barthelemy2022}
Barthelemy M, 2022 Class of models for random hypergraphs.
\newblock \emph{Phys. Rev. E} \textbf{106}, 064310.
\newblock (\doi{10.1103/PhysRevE.106.064310)}.

\bibitem{Krapivsky2023}
Krapivsky PL, 2023 Random recursive hypergraphs.
\newblock \emph{Journal of Physics A: Mathematical and Theoretical}
  \textbf{56}, 19, 195001.
\newblock (\doi{10.1088/1751-8121/accac0)}.

\bibitem{DEPANAFIEU2015}
Élie {de Panafieu}, 2015 Phase transition of random non-uniform hypergraphs.
\newblock \emph{Journal of Discrete Algorithms} \textbf{31}, 26--39.
\newblock (\doi{https://doi.org/10.1016/j.jda.2015.01.009)}.

\bibitem{KARONSKI2002}
Karoński M, Łuczak T, 2002 The phase transition in a random hypergraph.
\newblock \emph{Journal of Computational and Applied Mathematics} \textbf{142},
  1, 125--135.
\newblock (\doi{https://doi.org/10.1016/S0377-0427(01)00464-2)}.

\bibitem{Cooper2004}
Cooper C, 2004 The cores of random hypergraphs with a given degree sequence.
\newblock \emph{Random Structures \& Algorithms} \textbf{25}, 4, 353--375.
\newblock (\doi{https://doi.org/10.1002/rsa.20040)}.

\bibitem{Newman2009}
Ghoshal G, Zlati\ifmmode~\acute{c}\else \'{c}\fi{} V, Caldarelli G, Newman MEJ,
  2009 Random hypergraphs and their applications.
\newblock \emph{Phys. Rev. E} \textbf{79}, 066118.
\newblock (\doi{10.1103/PhysRevE.79.066118)}.

\bibitem{PARCZYK2015}
Parczyk O, Person Y, 2015 On Spanning Structures in Random Hypergraphs.
\newblock \emph{Electronic Notes in Discrete Mathematics} \textbf{49},
  611--619.
\newblock (\doi{https://doi.org/10.1016/j.endm.2015.06.083)}.

\bibitem{Kaminski2019}
Kamiński B, Poulin V, Prałat P, Szufel P, Théberge F, 2019 Clustering via
  hypergraph modularity.
\newblock \emph{PLOS ONE} \textbf{14}, 11, 1--15.
\newblock (\doi{10.1371/journal.pone.0224307)}.

\bibitem{Dewar2018}
Dewar M, Healy J, Pérez-Giménez X, Prałat P, Proos J, Reiniger B, Ternovsky
  K, 2018.
\newblock Subhypergraphs in non-uniform random hypergraphs.
\newblock (\doi{https://doi.org/10.48550/arXiv.1703.07686)}.

\bibitem{Asymptopia}
Spencer J, 2014 \emph{Asymptopia}, volume~71.
\newblock American Mathematical Soc.

\bibitem{Grzybowski2021}
Szymkuć S, Badowski T, Grzybowski BA, 2021 Is Organic Chemistry Really Growing
  Exponentially?
\newblock \emph{Angewandte Chemie International Edition} \textbf{60}, 50,
  26226--26232.
\newblock (\doi{https://doi-org.ezproxy.mis.mpg.de/10.1002/anie.202111540)}.

\bibitem{Lipkus2019}
Lipkus AH, Watkins SP, Gengras K, McBride MJ, Wills TJ, 2019 Recent Changes in
  the Scaffold Diversity of Organic Chemistry As Seen in the CAS Registry.
\newblock \emph{The Journal of Organic Chemistry} \textbf{84}, 21,
  13948--13956.
\newblock (\doi{10.1021/acs.joc.9b02111)}.

\bibitem{Hayashi2021}
Hayashi Y, 2021 Time and Pot Economy in Total Synthesis.
\newblock \emph{Accounts of Chemical Research} \textbf{54}, 6, 1385--1398.
\newblock (\doi{10.1021/acs.accounts.0c00803)}.

\bibitem{Ausiello1983}
Ausiello G, D'Atri A, Sacc\`{a} D, 1983 Graph Algorithms for Functional
  Dependency Manipulation.
\newblock \emph{J. ACM} \textbf{30}, 4, 752–766.
\newblock (\doi{10.1145/2157.322404)}.

\bibitem{WILD2017264}
Wild M, 2017 The joy of implications, aka pure Horn formulas: Mainly a survey.
\newblock \emph{Theoretical Computer Science} \textbf{658}, 264--292.
\newblock (\doi{https://doi.org/10.1016/j.tcs.2016.03.018)}.

\bibitem{Gallo1998}
Gallo G, Scutell{\`a} MG, 1998 Directed hypergraphs as a modelling paradigm.
\newblock \emph{Rivista di matematica per le scienze economiche e sociali}
  \textbf{21}, 1, 97--123.
\newblock (\doi{10.1007/BF02735318)}.

\bibitem{child2014molecular}
Child M, 2014 \emph{Molecular Collision Theory}.
\newblock Dover Books on Chemistry. Dover Publications.

\bibitem{Watts1998}
Watts DJ, Strogatz SH, 1998 Collective dynamics of `small-world' networks.
\newblock \emph{Nature} \textbf{393}, 6684, 440--442.
\newblock (\doi{10.1038/30918)}.

\bibitem{Albert2002}
Albert R, Barab\'asi AL, 2002 Statistical mechanics of complex networks.
\newblock \emph{Rev. Mod. Phys.} \textbf{74}, 47--97.
\newblock (\doi{10.1103/RevModPhys.74.47)}.

\end{thebibliography}

\end{document}